\documentclass[11pt]{amsart}
\usepackage[english]{babel}
\usepackage{tikz-cd,tikz}
\usetikzlibrary{matrix,arrows,decorations.pathmorphing}
\usepackage{accents}
\usepackage{graphicx,psfrag,geometry}
\usepackage{amssymb, dsfont, amsthm, amsmath, bm, mathrsfs, yfonts,natbib}%xspace
\newtheorem{thm}{Theorem}[section]

\newtheorem{lem}[thm]{Lemma}
\newtheorem{prop}[thm]{Proposition}
\newtheorem{defn}[thm]{Definition}
\theoremstyle{remark}
\newtheorem{rem}[thm]{Remark}
\newtheorem{example}[thm]{Example}

\newcommand{\set}[1]{\left\{#1\right\}}
\newcommand{\trr}{\triangleright}
\newcommand{\rrt}{\triangleleft\,}
\newcommand{\abs}[1]{\left\vert#1\right\vert}
\newcommand{\ass}{\stackrel{\textup{\tiny def}}{=}}
\newcommand*\mycirc[1]{\!\!\begin{tikzpicture}[baseline=(C.base)]\node[draw,circle,inner sep=1.5pt](C) {\scalebox{0.8}{$#1$}};\end{tikzpicture}\!\!}
\newcommand{\ooplus}{\mycirc{\trr}}
\newcommand{\oominus}{\mycirc{\rrt}}
\newcommand{\oodiamond}{\mycirc{\diamond}}
 \newcommand{\dU}{ \, \raisebox{1.25pt}{\scalebox{.8}{$\coprod$}} \; }

\newcounter{saveenumerate} %intertext in enumerated lists
\makeatletter
\newcommand{\enumeratext}[1]{%
\setcounter{saveenumerate}{\value{enum\romannumeral\the\@enumdepth}}
\end{enumerate}
#1
\begin{enumerate}
\setcounter{enum\romannumeral\the\@enumdepth}{\value{saveenumerate}}%
}
\makeatother

\begin{document}

\title[Tangle Machines \textrm{II}:\ \, Invariants]
{Tangle Machines \textrm{II}:\ \, Invariants}

\author{Daniel Moskovich}%
\address{Division of Mathematics, School of Physical and Mathematical Sciences, Nanyang Technological University, 21 Nanyang Link, Singapore 637371}%
\email{dmoskovich@gmail.com}%

\author{Avishy Y. Carmi}
\address{Department of Mechanical Engineering, Ben-Gurion University of the Negev, Beer-Sheva 8410501, Israel}
\email{avishycarmi@gmail.com}

\keywords{diagrammatic models, natural computing, recursion, adiabatic quantum computing, information theory, cybernetics, networks, knot theory, reidemeister moves}

\date{10th of April, 2014}

%\address{Division of Mathematical Sciences, Nanyang Technological University\\
%SPMS-MAS-03-01, 21 Nanyang Link Singapore 637371 }%
%\email{dmoskovich@ntu.edu.sg}%
%\urladdr{http://www.sumamathematica.com/}

\begin{abstract}
The preceding paper constructed tangle machines as diagrammatic models, and illustrated their utility with a number of examples. The information content of a tangle machine is contained in characteristic quantities associated to equivalence classes of tangle machines, which are called \emph{invariants}. This paper constructs invariants of tangle machines. Chief among these are the \emph{prime~factorizations} of a machine, which are essentially unique. This is proven using low dimensional topology, through representing a colour-suppressed machine as a diagram for a network of jointly embedded spheres and intervals in $4$--space. The \emph{complexity} of a tangle machine is defined as its number of prime factors.
\end{abstract}
%\date{}%
%\dedicatory{}%
%\commby{}%
% ----------------------------------------------------------------
%\label{firstpage}
\maketitle

% ----------------------------------------------------------------

\section{Introduction}

The prequel to this paper defined \emph{tangle machines}, a low dimensional topological formalism for causality, computation, and information. Equivalent machines are considered `globally the same', meaning that one can be perfectly reproduced from another, but perhaps not `locally the same'. That paper provided examples of machines modeling recursion and Markov chains, networks of adiabatic quantum computations, and networks of distributed information processing. In each example three equivalent machines were presented, one `optimal', one `suboptimal', and one `abstract'.

The goal of the present paper is to extract information from machines in the form of \emph{machine invariants}. Invariants  are numbers, polynomials, and other well-understood mathematical objects associated to equivalence classes of machines. \emph{Information invariants} are those invariants $v$ such that, if $M_1\begin{minipage}{8pt}\includegraphics[width=8pt]{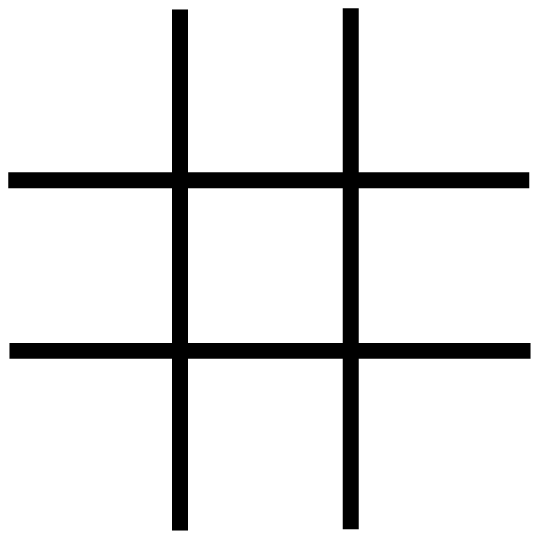}\end{minipage} M_2$ is the connect sum of $M_1$ with $M_2$, then

\begin{equation}
v(M_1\begin{minipage}{8pt}\includegraphics[width=8pt]{hash.eps}\end{minipage} M_2)=v(M_1)+v(M_2).
\end{equation}

Information invariants of tangle machines capture information theoretical quantities associated to the machine. These include:

\begin{itemize}
\item The relative influence of one part of the machine on another (Sections~\ref{SSS:LinkingGraph} and \ref{SSS:SpectralGraph}).
\item The number of nontrivial interactions contained in the machine (Section~\ref{SSS:InteractionNumber}).
\item The number of independently functioning components (\emph{factors}) of a machine (Section~\ref{S:Complexity}).
\item The maximal amount of information that a machine can contain (Section~\ref{SSS:ShannonCapacity}).
%\item At least for certain classes of quantum systems the \emph{entropy} of the machine %(SECTION).
\end{itemize}

Low dimensional topology provides a toolbox with which to prove that information invariants, especially nontrivial interaction number and capacity, are indeed invariants, \textit{i.e.} that they take the same value for equivalent machines. The authors do not know how invariance may otherwise be proven, for indeed the only known proofs for parallel statements in knot theory are topological in nature.

This paper is organized as follows. In Section~\ref{S:TopPrelim} we recall from the prequel the relevant tangle machine definitions, and we provide a number of low dimensional topological preliminaries to subsequent sections. In Section~\ref{SS:SphereInterval} we reveal machines to be diagrams for networks of spheres and intervals jointly embedded in standard Euclidean $4$--space $\mathds{R}^4$. We prove a Reidemeister Theorem for machines, and provide two alternative diagrammatic formalisms for machines, via \emph{Roseman diagrams} and via hybrid \emph{Rosemeister diagrams}. Section~\ref{SS:SimpleInvariants} discusses various relatively simple information invariants of machines, and finally in Section~\ref{S:Complexity} we discuss machine complexity, that is the number of \emph{prime factors} of a machine (the maximal number of its nontrivial connect summands). The unique prime factorization theorem of that section has a parallel for classical knots, but not for virtual knots or for w-knots. It states roughly that every nontrivial factorization of a non-split machine has a unique maximal refinement, up to \emph{unit factors} each of which contain only one colour.

For ease of exposition, in this paper we assume throughout that the rack operation $\trr$ is the same at all crossings.

\section{Preliminaries}\label{S:TopPrelim}

\subsection{Machines}\label{SS:MachinePrelim}

In order to make this paper logically self-contained, we recall some definitions from the prequel.

\begin{defn}[Rack; Quandle]
\label{D:Rack}\hfill
\begin{itemize}
\item A \emph{rack} is a set $Q$ equipped with a binary operation $\trr$ such that $\trr\, z$ is an automorphism of $Q$ for all $z\in Q$, and such that $Q$ is closed under the inverse operation $\rrt$ of $\trr$.
 \item A rack all of whose colours are idempotent, \textit{i.e.} $x\trr x=x$ for all $x\in Q$, is called a \emph{quandle}.
\end{itemize}
\end{defn}

\begin{defn}[Tangle machines]\label{D:TangleMachine}
A \emph{tangle machine} $M$ is a triple $M\ass (G,\bm{\phi},\rho)$ consisting of:
\begin{itemize}
\item A disjoint union of directed path graphs $A_1,\ldots, A_k$ (\emph{open processes}) and directed cycles $C_1,\ldots,C_l$ (\emph{closed processes}),
\begin{equation}
G\,\ass\, \left(A_1\dU A_2\dU \cdots\dU A_k\right)\dU \left(C_1\dU C_2\dU\cdots\dU C_l\right),
\end{equation}
The graph $G$ is called the \emph{underlying graph} of $M$. Vertices of $G$ are called \emph{registers}.

\item A partially-defined \emph{interaction function}
\begin{equation}
\bm{\phi}\ass (\phi,\mathrm{sgn})\colon\, E(G)\to V(G)\times\{+,-\}
\end{equation}
%a permutation $\sigma\in \mathcal{S}_k$ (needed to make the GLUING well-defined),
\item A \emph{colouring function} $\rho$ from $V(G)$ to a rack $Q$ such that, if $v$ and $w$ are vertices in $M$ and if $e$ is an edge from $v$ to $w$, we have:
\begin{equation}\label{E:Compatability}\left\{
    \begin{array}{ll}
      \rho(v)\trr \rho(\phi(e))=\rho(w), & \hbox{If $\mathrm{sgn}(e)=+$;} \\
      \rho(v)\triangleleft \rho(\phi(e))=\rho(w), & \hbox{if $\mathrm{sgn}(e)=-$;}\\
      \rho(v)=\rho(w) & \hbox{if $e\notin \textrm{Domain}(\bm{\phi})$.}
    \end{array}
  \right.
\end{equation}
\end{itemize}
%where $\rrt$ indicates the inverse operation to $\trr$.
\end{defn}

If $Q$ is a quandle then a machine $M$ is said to be a \emph{quandle machine}. Conversely, we refer to $M$ as a \emph{rack machine} when we wish to stress that $Q$ is not a quandle.

Two machines $M_1$ and $M_2$ are considered \emph{equivalent} if they are related by an automorphism of $Q$ together with a finite sequence of the following \emph{Reidemeister moves}:

 \begin{equation}
\begin{tikzcd}[row sep=1em,
column sep=0.8em]
 x \rar[dash] & \ooplus \rar & x\triangleright y \rar[dash] & \oominus \rar & x\\
& & y \arrow[dash, dashed]{ur} \arrow[dash, dashed]{ul} & &
\end{tikzcd}\quad \stackrel{R2}{\rule{0pt}{5pt}\longleftrightarrow} \quad\  \begin{tikzcd}[row sep=1em, column sep=0.5em]
x \arrow{rr} & & x \arrow{rr} & & x \\ & & y & &
\end{tikzcd} \quad \stackrel{R2}{\rule{0pt}{5pt}\longleftrightarrow} \quad\  \begin{tikzcd}[row sep=1em,
column sep=0.8em]
 x \rar[dash] & \oominus \rar & x\rrt y \rar[dash] & \ooplus \rar & x\\
& & y \arrow[dash, dashed]{ur} \arrow[dash, dashed]{ul} & &
\end{tikzcd}
\end{equation}

\begin{equation}
\resizebox{.4\hsize}{!}{
\begin{tikzcd}[ampersand replacement=\&,row sep=1em, column sep=0.4em]
(x_1 \trr z) \trr (y \trr z) \& \& \cdots \& \& (x_k \trr z) \trr (y \trr z) \\
\ooplus \arrow{u} \& \& y \trr z \arrow[dash, dashed]{ll} \arrow[dash, dashed]{rr} \& \& \ooplus \arrow{u} \\
x_1 \trr z \arrow[dash]{u} \& \& \ooplus \arrow{u} \& y \arrow[dash]{l} \& x_k \trr z \arrow[dash]{u} \\
\ooplus \arrow{u} \& \& z \arrow[dash, dashed]{u} \arrow[dash, dashed]{ll} \arrow[dash, dashed]{rr} \& \& \ooplus \arrow{u} \\
x_1 \arrow[dash]{u} \& \& \cdots \& \& x_k \arrow[dash]{u}
\end{tikzcd}} \stackrel{R3}{\longleftrightarrow}
\resizebox{.4\hsize}{!}{\begin{tikzcd}[ampersand replacement=\&,row sep=1em, column sep=0.4em]
(x_1 \trr y) \trr z \& \& \cdots \& \& (x_k \trr y) \trr z \\
\ooplus \arrow{u} \& \& y \arrow[dash]{d} \arrow[dash, dashed, bend right=10]{lldd} \arrow[dash, dashed, bend left=10]{rrdd} \& \& \ooplus \arrow{u} \\
x_1 \trr y \arrow[dash]{u} \& \; \; \; \; \& \ooplus \arrow{r} \& y \trr z \& x_k \trr y \arrow[dash]{u} \\
\ooplus \arrow{u} \& \& z \arrow[dash, dashed]{u} \arrow[dash, dashed, bend left=10]{lluu} \arrow[dash, dashed, bend right=10]{rruu} \& \& \ooplus \arrow{u} \\
x_1 \arrow[dash]{u} \& \& \cdots \& \& x_k \arrow[dash]{u}
\end{tikzcd}}
\end{equation}

If $M$ is a quandle machine, we admit also the following move:

 \begin{equation}\label{E:R1}
\begin{tikzcd}[row sep=1em,
column sep=0.8em]
 x \rar[dash] & \oodiamond\arrow[dash,dashed, bend left=35]{l} \rar & \
\end{tikzcd}
\stackrel{R1}{\longleftrightarrow} \quad \begin{tikzcd}[row sep=1em,
column sep=0.8em]
 x \arrow{rr} & &  \
\end{tikzcd} \quad\quad\text{and}\quad\ \ \begin{tikzcd}[row sep=1em,
column sep=0.8em]
 \  \arrow{rr} & &  x
\end{tikzcd}\ \stackrel{R1}{\longleftrightarrow} \!\begin{tikzcd}[row sep=1em,
column sep=0.8em]
 \  \rar[dash] & \oodiamond \rar & x \arrow[dash,dashed, bend left=35]{l}
\end{tikzcd}
\end{equation}

The following move is called \emph{stabilization}, where one of the registers on the LHS must lie outside the image of $\phi$:

\begin{equation}\label{E:FalseStabilization}
\begin{tikzcd}[row sep=1em,
column sep=1em]
x \arrow{rr} & & x
\end{tikzcd}\quad\longleftrightarrow\quad \begin{tikzcd}[row sep=1em]
x
\end{tikzcd}
\end{equation}
If both registers on the LHS are in the image of $\phi$ then the above move is called \emph{false stabilization}.

Two machines $M_1$ and $M_2$ that are related by automorphisms of $Q$ on their connected components, a finite sequence of Reidemeister moves (and (de)stabilizations) are said to be \emph{(stably) equivalent}.

Another diagrammatic formalism for machines is perhaps easier for a human to work with. We redraw an interaction as a \emph{crossing}:

\begin{equation}\label{E:kebab}
\begin{tikzcd}[row sep= 3em, column sep=0.8em]
x_1 \rar[dash]& \oodiamond \rar & x_1\diamond y & x_2 \rar[dash]& \oodiamond \rar & x_2\diamond y &\cdots & x_k \rar[dash]& \oodiamond \rar & x_k\diamond y\\
\ & \ & \ & \ & \ & y \arrow[dash, dashed]{ullll} \arrow[dash, dashed]{ul} \arrow[dash, dashed]{urrr}
\end{tikzcd}\hspace{-10pt} \quad \quad \begin{minipage}{0.8in}\psfrag{a}[c]{$x_1$}\psfrag{x}[l]{$x_2$}\psfrag{y}[c]{$x_k$}\includegraphics[width=0.8in]{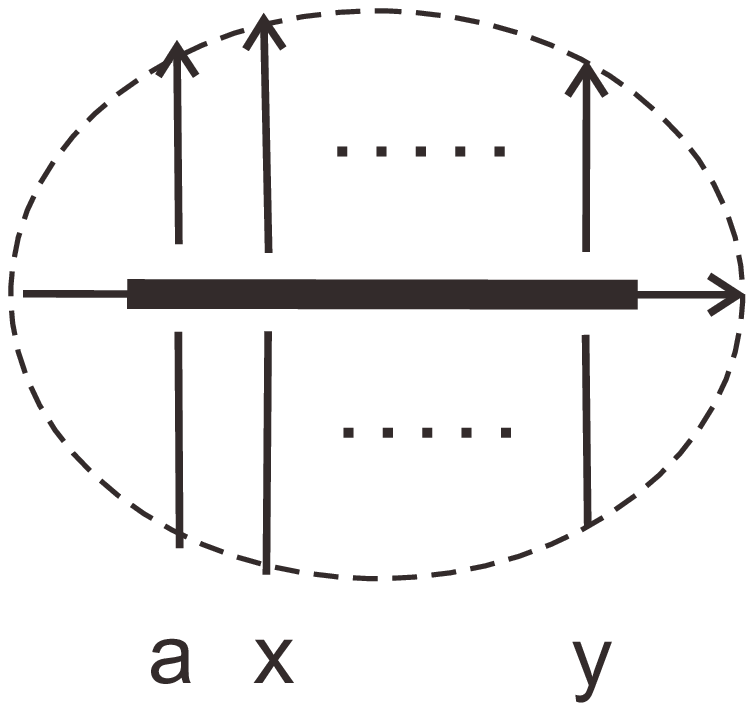}\end{minipage}
\end{equation}

Concatenating as required, we obtain a \emph{Reidemeister diagram} for our machine $M$. Two Reidemeister diagrams are said to be \emph{stably equivalent} if they are related by composing an automorphism of $Q$ with the colourings of their connected components, together with a finite sequence of the local moves listed in  Figure~\ref{F:local_moves_machines}. If stabilization is not used, then the two Reidemeister diagrams are said to be \emph{equivalent}. Stable equivalence classes of machines and of Reidemeister diagrams coincide.

\begin{figure}
\psfrag{T}[c]{\small \emph{VR1}}\psfrag{R}[c]{\small \emph{VR2}}\psfrag{S}[c]{\small \emph{VR3}}
\psfrag{Q}[c]{\small \emph{SV}}\psfrag{D}[c]{\small \emph{R1}}\psfrag{A}[c]{\small \emph{R2}}\psfrag{B}[c]{\small \emph{R3}}\psfrag{C}[c]{\small \emph{UC}}
\psfrag{X}[c]{\small \emph{ST}}
\includegraphics[width=0.9\textwidth]{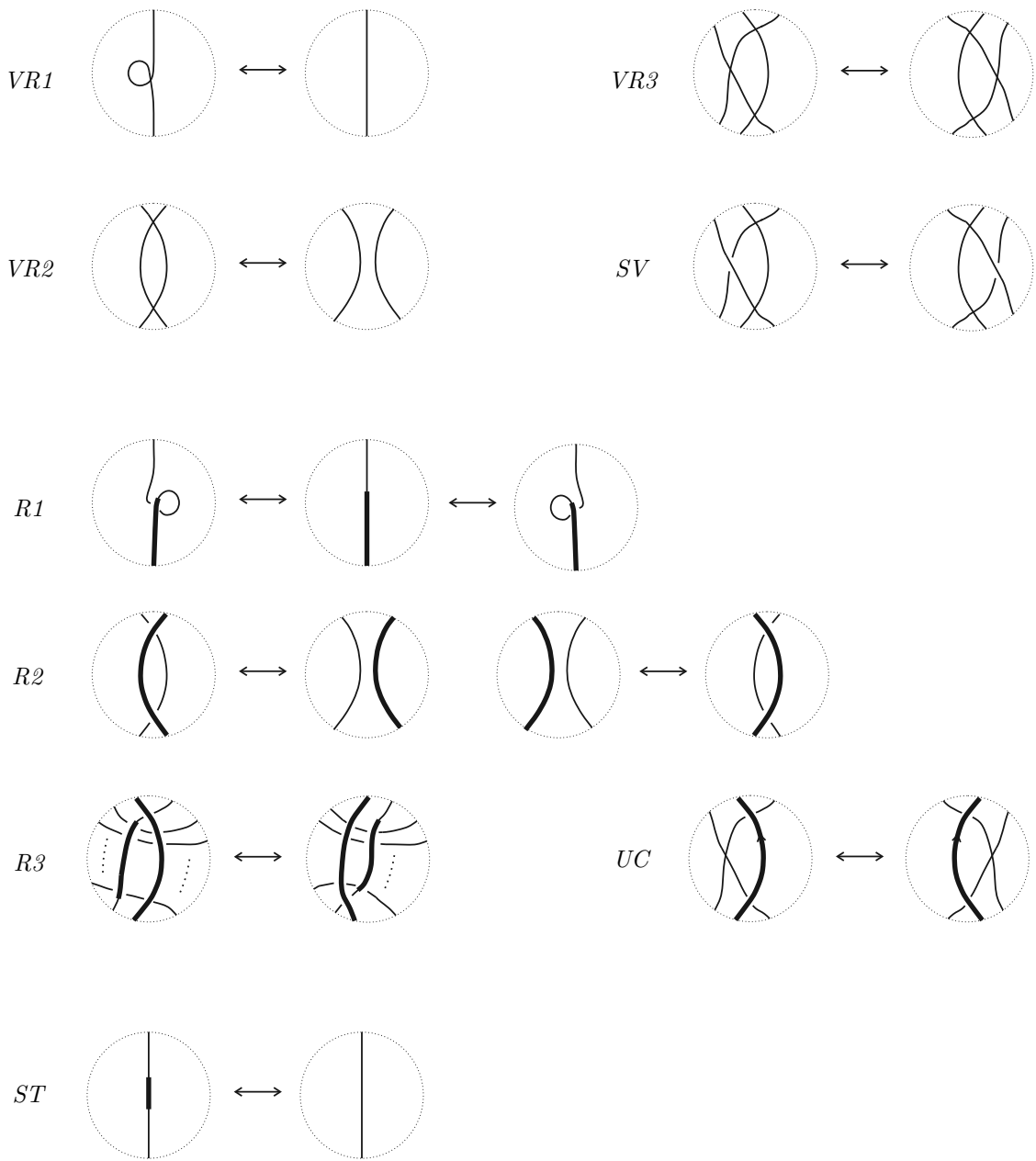}
\caption{\label{F:local_moves_machines}Local moves for machines, valid for any orientations of the strands. The R1 move is valid for quandle machines but not for rack machines.}
\end{figure}

Recall that a machine $M\ass (G,\bm{\phi},\rho)$ is a \emph{connect sum} of $M_1\ass (G,\bm{\phi}_1,\rho_1)$ and $M_2\ass (G,\bm{\phi}_2,\rho_2)$ if, writing $A_1\ass \mathrm{Domain}(\bm{\phi}_1)$ and $A_2\ass \mathrm{Domain}(\bm{\phi}_2)$, we have $A_1\cap A_2=\emptyset$ and $\mathrm{Domain}(\bm{\phi})=A_1\cup A_2$, with $\rho(r)=\rho_1(r)$ for all $r\in A$ and $\rho(r)=\rho_2(r)$ for all $r\in B$. In this case we write $M\ass M_1\begin{minipage}{8pt}\includegraphics[width=8pt]{hash.eps}\end{minipage} M_2$.

\begin{multline}
\resizebox{.4\hsize}{!}{
\begin{tikzcd}[ampersand replacement=\&,row sep=1em, column sep=1em]
\& \& \ooplus  \arrow[dash, dashed]{dd} \arrow{rr} \& \& y \arrow{rr}  \arrow[dash,dashed]{ddlll} \& \& y \arrow[dash]{rr}  \& \&  \ooplus  \arrow{drr} \& \& \\
x\trr y \arrow[dash]{urr} \arrow[dash, dashed, bend left=10]{drrr} \& \& \& \& \& \& \& \& \& \& y  \arrow[dash]{dl} \\
\&  \ooplus \arrow{ul} \rar[dash] \& x \&  \ooplus \lar \& y \arrow[dash]{l} \& \& y\arrow{ll}  \&  \ooplus \lar  \& y \arrow[dash]{l} \&  \ooplus \lar \&
\end{tikzcd}}\ \begin{minipage}{8pt}\includegraphics[width=8pt]{hash.eps}\end{minipage}\
\resizebox{.4\hsize}{!}{\begin{tikzcd}[ampersand replacement=\&,row sep=1em, column sep=1em]
\& \& \ooplus  \arrow{rr} \& \& y \arrow{rr} \& \& y \arrow[dash]{rr} \arrow[dash,dashed]{ddrrr} \& \&  \ooplus \arrow[dash,dashed]{dd} \arrow{drr} \& \& \\
y \arrow[dash]{urr} \& \& \& \& \& \& \& \& \& \& y\trr x  \arrow[dash]{dl} \\
\&  \ooplus \arrow{ul} \rar[dash] \& y \&  \ooplus \lar \& y \arrow[dash]{l} \& \& y\arrow{ll}  \&  \ooplus \arrow[dash,dashed, bend left=10]{urrr} \lar  \& x \arrow[dash]{l} \&  \ooplus \lar \&
\end{tikzcd}}\\ =\
\resizebox{.4\hsize}{!}{\begin{tikzcd}[ampersand replacement=\&,row sep=1em, column sep=1em]
\& \& \ooplus  \arrow[dash, dashed]{dd} \arrow{rr} \& \& y \arrow{rr} \arrow[dash,dashed]{ddlll}\& \& y \arrow[dash]{rr} \arrow[dash,dashed]{ddrrr} \& \&  \ooplus \arrow[dash,dashed]{dd} \arrow{drr} \& \& \\
x\trr y \arrow[dash]{urr} \arrow[dash, dashed, bend left=10]{drrr} \& \& \& \& \& \& \& \& \& \& y\trr x  \arrow[dash]{dl} \\
\&  \ooplus \arrow{ul} \rar[dash] \& x \&  \ooplus \lar \& y \arrow[dash]{l} \& \& y\arrow{ll}  \&  \ooplus \arrow[dash,dashed, bend left=10]{urrr} \lar  \& x \arrow[dash]{l} \&  \ooplus \lar \&
\end{tikzcd}}
\end{multline}

The converse of connect sum is \emph{cancellation}. To \emph{cancel} a factor $N=(H,\bm{\phi}_H,\rho_H)$ in $M=(G,\bm{\phi},\rho)$ is to replace $M$ by a machine $M-N\ass (G,\bm{\phi}_{G-H},\rho^H)$  where the $\rho^H$ satisfies $\rho^H(r)=\rho(r)$ for all $r\in G-H$. Here, $\bm{\phi}_{G-H}$ denotes the restriction of $\bm{\phi}$ to $G-H$.

\subsection{Knotted surfaces}\label{SSS:BrokenSurfaces}

In Section~\ref{SS:SphereInterval}, colour-suppressed Reidemeister diagrams of machines are conceived of as diagrams for jointly embedded networks of spheres and intervals. In this section we recall the rudiments of the classical theory of knotted surfaces.

Embeddings of $k$--dimensional objects in $k+2$--dimensional Euclidean space generalize classical knots. The $k=2$ case is the case of knotted surfaces in Euclidean $\mathds{R}^4$. Knotted surfaces are traditionally described by \emph{broken surfaces diagrams}, which are analogous to knot diagrams, and which we shall also call \emph{Roseman diagrams}. A reference for these is \citep{CarterKamadaSaito:04}, to which we refer the reader for details.

Let $K\colon\, \Sigma\to \mathds{R}^4$ be a smooth embedding in Euclidean $\mathds{R}^4$ of a closed surface $\Sigma$. Choose and fix a vector $v$ in $\mathds{R}^4$, which we once and for all identify with the $t$--axis. Its orthogonal complement is a hyperplane $H\subset\mathds{R}^4$, which is identified with $\mathds{R}^3$ with the $(x,y,z)$--axes. Project $K$ onto $H$ via a projection $\pi$. Generically, the singular points of the projection will be double-points, triple-points, and branch points. Neighbourhoods of each of these are as given in Figure~\ref{F:BranchPoints}.

\begin{figure}
\centering
\includegraphics[width=4.5in]{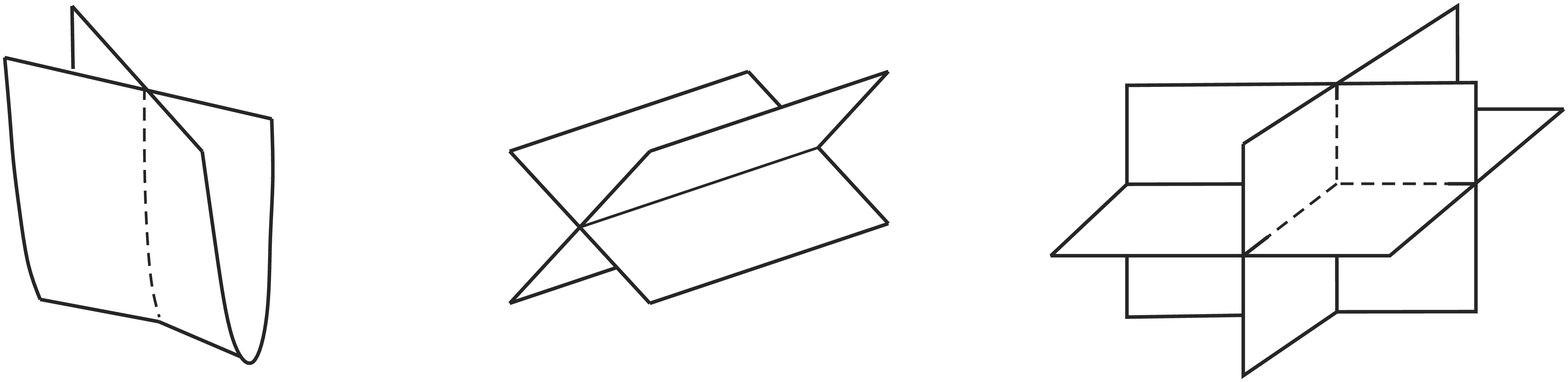}
\caption{\label{F:BranchPoints}Neighbourhoods of singular points of a generic $3$--dimensional projection of a knotted surface in $\mathds{R}^4$.}
\end{figure}

If the $t$--coordinate of a neighbourhood $N$ of a point $p$ in $K$ is greater than the $t$--coordinate of a neighbourhood $N^\prime$ of a point $p^\prime$ in $K$, and if $\pi(N)\cap\pi(N^\prime)\neq\emptyset$, break $\pi(N)$. This parallels the breaking, in the knot diagram case, of the line in the knot diagram whose pre-image is further from the projection plane into two undercrossing arcs.

The analogue to the Reidemeiser Theorem for knotted surfaces \citep{HommaNagase:85, CarterSaito:93, Roseman:98} reads as follows:

\begin{thm}[Roseman Theorem]
Two smooth embeddings $K_1,K_2$ of a closed surface are ambient isotopic if and only if any broken surfaces diagram $D_1$ of $K_1$ is related to any broken surfaces diagram $D_2$ of $K_2$ by a finite sequence of \emph{Roseman moves}, as shown in Figure~\ref{F:Roseman1}.
\end{thm}

\begin{figure}[htb]
\centering
\includegraphics[width=0.95\textwidth]{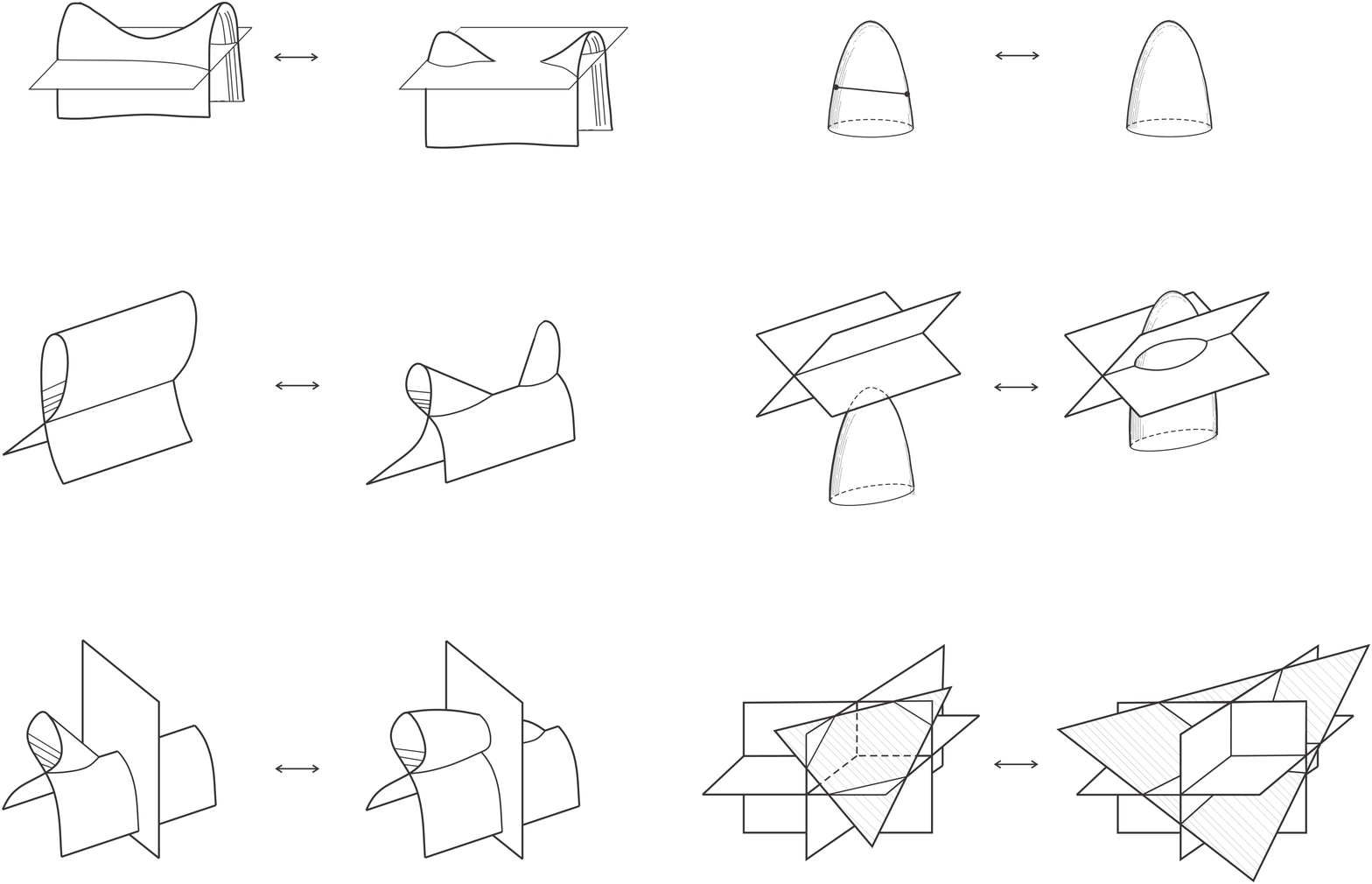}
\caption{\label{F:Roseman1}The Roseman moves. These are valid for all ways of breaking the above surfaces.}
\end{figure}

To illustrate our diagrammatic language, Figure~\ref{F:WKnot} presents some diagrams of a ribbon torus knot, the set of which is conjecturally in bijective correspondence with the set of w-knots.

\begin{figure}
\centering
\includegraphics[width=0.5\textwidth]{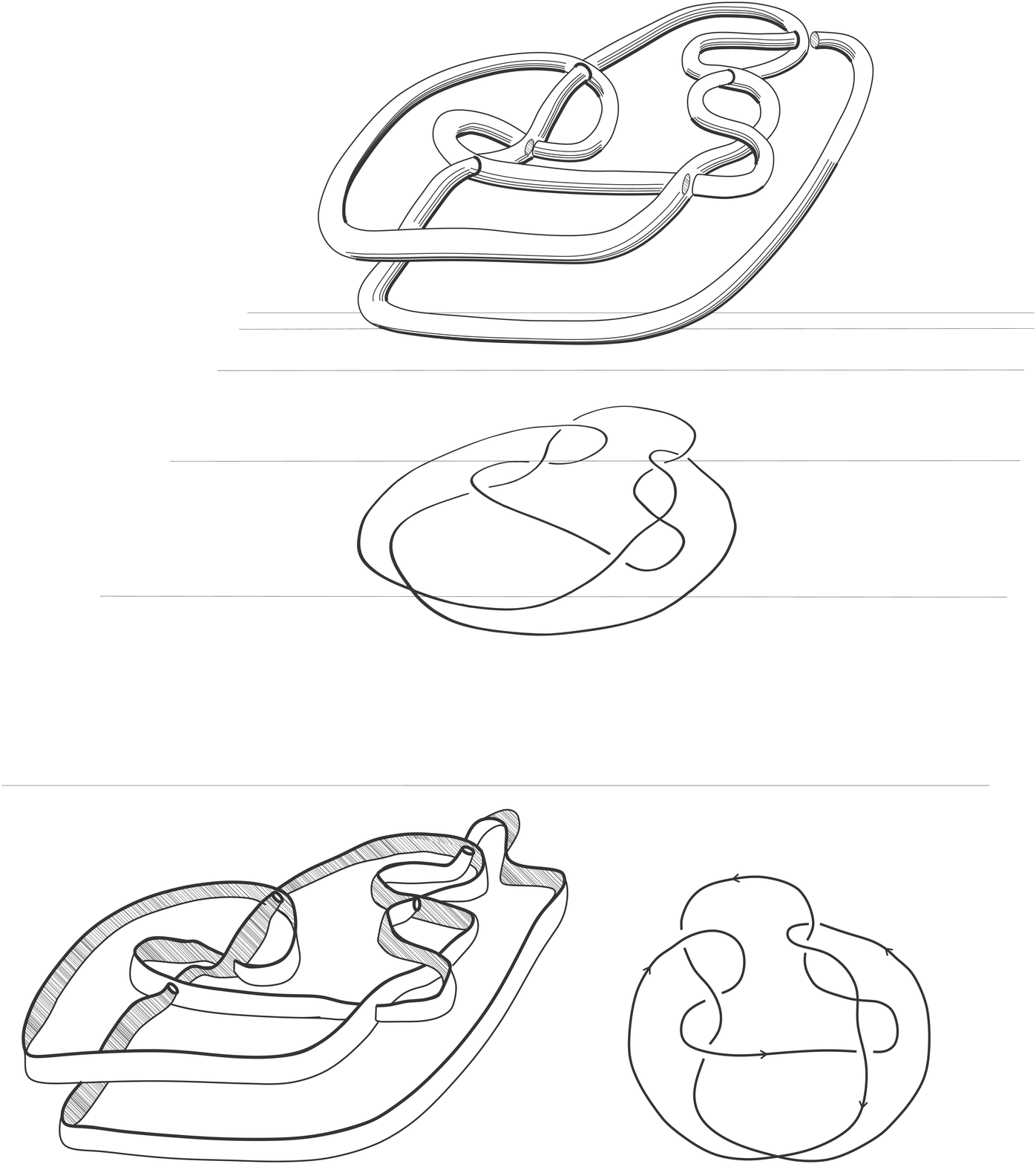}
\caption{\label{F:WKnot} A Reidemeister diagram, a Roseman diagram, and a Rosemeister diagram of a ribbon torus knot. See Section~\ref{SS:Rosemeister} for the definition of the latter, in the case of sphere-and-interval tangles.}
\end{figure}

\section{Sphere-and-interval tangles}\label{SS:SphereInterval}

In this section, we exhibit the topological nature of machines. Topology is well~suited to describe information-preserving modifications; this section provides the conceptual underpinning for why we expect machines to be an effective tool to simplify (or to complicate) descriptions of information~transfer between interacting processes while preserving the information content that we are interested in.

We exhibit a topological `lift' of our diagrammatic notation to networks of spheres and intervals tangled together in $4$--space. These are drawn via their Roseman diagrams. It is here that we prove a Reidemeister Theorem for machines, Theorem~\ref{T:MachineReidemeisterTheorem}. Finally, in Section~\ref{SS:Rosemeister}, we define a compromise between Reidemeister diagrams and Roseman diagrams which we call \emph{Rosemeister diagrams}. Rosemeister diagrams also exist for $w$--knots, and we explain the relationship between machines and $w$--knots in Section~\ref{SSS:wknotsrel}.

In this section, we reinterpret the colour-suppressed Reidemeister diagram of a machine as a planar projection of network of spheres and intervals, knotted in $4$--space. This construction reveals a colour-suppressed Reidemeister diagram of a machine as arising from a projection to a plane $P$ of a tangled system of $2$-spheres $S^2$ and intervals, equipped with colouring information. Colour-suppressed machines thus correspond to topological objects.

Our construction is similar to the `balloons and hoops' construction of Bar-Natan \citep{BarNatan:13}, although different knotted objects are being described.

In this section, the words `up' and `down' are to be interpreted with respect to the right-hand convention.

\subsection{Constructing sphere-and-interval tangles representing machines}\label{SS:sandsquandle}

%We first describe a sphere-and-interval tangle for a quandle machine $M$. In Section \ref{SS:sandsrack} we modify this construction for rack machine.
Recall the Roseman diagrams of Section \ref{SSS:BrokenSurfaces}.

Begin by constructing a local model for a single interaction, consisting of a single over-strand $A$ with $k$ strands passing up through it and $l$ strands passing down through it. Consider a $2$--sphere in Euclidean $\mathds{R}^4$:

\begin{equation}
S\ass \left\{\left.\rule{0pt}{11pt}(x,y,z,0)\in \mathds{R^4}\right|\, \sqrt{x^2+y^2+z^2}=1\right\}.
\end{equation}

Orient $S$ according to the right-hand convention, \textit{i.e.} so that the intersection of $S$ with the $XY$--plane is oriented counterclockwise. The sphere $S$ represents the over-strand $A$.

\begin{rem}
If we want to be rigourous, then a different embedding of the sphere is to be preferred. We choose:
\begin{equation}\label{E:etrog}
S\ass \left\{\left.\rule{0pt}{11pt}(\sigma(z)x,\sigma(z)y,z,0)\in \mathds{R^4}\right|\, -1\leq x\leq 1;\ \sqrt{y^2+z^2}=1\right\}.
\end{equation}
\noindent where $\sigma\colon\, [-1,1]\to [0,1]$ is a modified logistic function $\frac{1}{2}+\frac{1}{2}\tanh\left(\tan(\frac{\pi}{2}z)\right)$ for $x\in (-1,1)$ and with $\sigma(\pm 1)\ass 0$. This is because we need to define a sphere-and-interval tangle to be a stratified space in order for smooth ambient isotopy of such objects to be well-defined \citep{GoreskyMacPherson:88}.

For ease of exposition we'll pretend that $S$ is parameterized as a sphere, but it's actually parameterized as Equation \ref{E:etrog}.
\end{rem}

Consider now parameterized intervals $l_j^t$ with $t\in[-2,2]$ so that:
\begin{equation}
l_j^t \ass \left\{
             \begin{array}{ll}
               (\frac{j+1}{l+k+2},t,0,1), & \hbox{for $0<j\leq k$;} \\
               (\frac{j+1}{l+k+2},-t,0,1)\rule{0pt}{12pt}, & \hbox{for $k<j\leq l+k$.}
             \end{array}
           \right.
\end{equation}
Thus, an under-strand passing ``up'' through $A$ corresponds to an interval passing up through $S$, and vice versa. Finally, adjoin two parameterized intervals $l_A^+\ass (0,0,1+t,0)$ and  $l_A^-\ass (0,0,-2+t,0)$ of length $1$. The figure which we have constructed, which we have drawn in Figure~\ref{F:kebaby}, lies inside a $4$--dimensional $4\times 4\times 4\times 4$ cube $B$.
\begin{figure}[htb]
\begin{minipage}[t]{1\linewidth}
\psfrag{s}[c]{\small $S$}
\psfrag{a}[c]{\small $l_A^+$}
\psfrag{b}[c]{\small $l_A^-$}
\centering
    \includegraphics[width=0.3\textwidth]{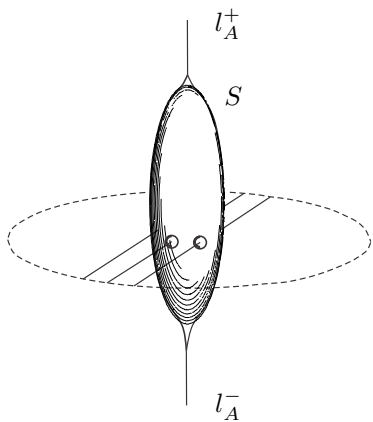}
\end{minipage}
\caption{\label{F:kebaby} A sphere-and-interval tangle corresponding to a single interaction.}
\end{figure}

The next step is to concatenate. At this point, the $4$--dimensional figure that we have construction, which consists of spheres $S_1,S_2,\ldots, S_N$ and of intervals, lies inside a collection of $4\times 4\times 4\times 4$ cubes $B_1,B_2,\ldots,B_N$. We index these so that $S_i$ lies inside $B_i$ for all $1\leq i\leq N$., and embed the cubes disjointedly in $\mathds{R}^4$. Concatenate by connecting endpoints of intervals on the boundaries of the cubes (these are endpoints of $l_j$ intervals and of $l_A$ intervals) to one another, corresponding to how the registers which represent them connect with one another in $M$. The embedding should be chosen so that the concatenation of two smooth embedded intervals is again a smooth embedded interval. Line segments added for the purpose of concatenation should lie entirely outside $B_1,B_2,\ldots,B_k$, and should not intersect.

Finally, for each intersection $p$ of one of the intervals $l_j$ or $l_A^\pm$ with the boundary $\partial B$ of a cube $B$, endpoints of $l_j$ intervals or of $l_A$ intervals which have not been used for concatenation embed a ray into $\mathds R^4$ so that its endpoint maps to $p$ and its open end diverges to $\infty$, requiring again that it not intersect any of the other geometric objects which we have placed. These rays correspond to endpoints of the machine $M$.

We have obtained a geometric figure in $4$--dimensional space, which we call a \emph{sphere-and-interval tangle}. See Figure~\ref{F:tangle_4dnew}.

\begin{figure}[htb]
\begin{minipage}[t]{1\linewidth}
\centering
    \includegraphics[width=0.7\textwidth]{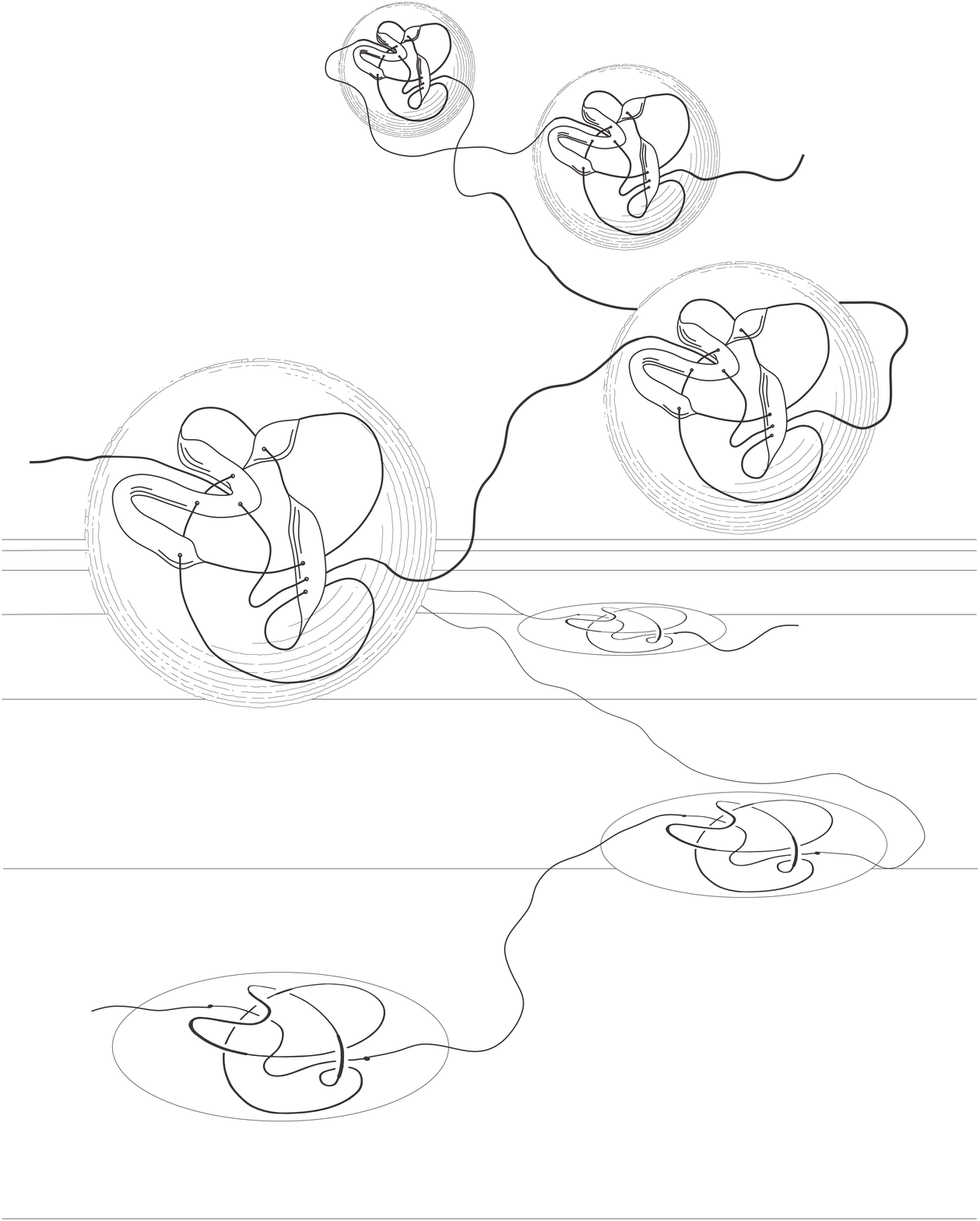}
\end{minipage}
\caption{\label{F:tangle_4dnew}The sphere-and-interval tangle above is represented by a tangle machine, appearing as the `shadow' on the plane underneath, together with a collection of decomposing spheres (discussed later).}
\end{figure}

The above discussion has associated a sphere-and-interval tangle to each colour-suppressed Reidemeister diagram. To translate back from sphere-and-interval tangles to colour-suppressed Reidemeister diagrams, represent each sphere as an over-strand and each line through it as an under-strand, and concatenate as required.

Figure~\ref{F:FlyingRings} illustrates the same sphere-and-interval tangle in various dimensions. In the top representation, that is $4$--dimensional, each ring represents a $3$--dimensional slice, with the $3$--dimensional slices being $1$ `time' unit apart. One of the strands is colour-coded red, another black, and the third blue. The coloured ring represents time zero, and $t-a$ represents time $a$, with the colour of the characters representing the colour of the strand. Thus, a blue $t-7$ and a black $t-5$ indicates that the black ring passes at time $5$, and the blue at time $7$, after it (hence they don't collide). Thus blue passes through black ``north to south'', as shown in the $3$--dimensional representation. The same interpretation holds also for the left `crossing'.

\begin{figure}[htb]
\psfrag{t5}[c]{\small $t-5$}
\psfrag{t7}[c]{\color{blue} \small $t-7$}
\psfrag{t}[c]{\color{blue} \small $t-11$}
\psfrag{t12}[c]{\small $t-12$}
 \psfrag{t13}[c]{\color{red} \small $t-13$}
\centering
\includegraphics[width=0.95\textwidth]{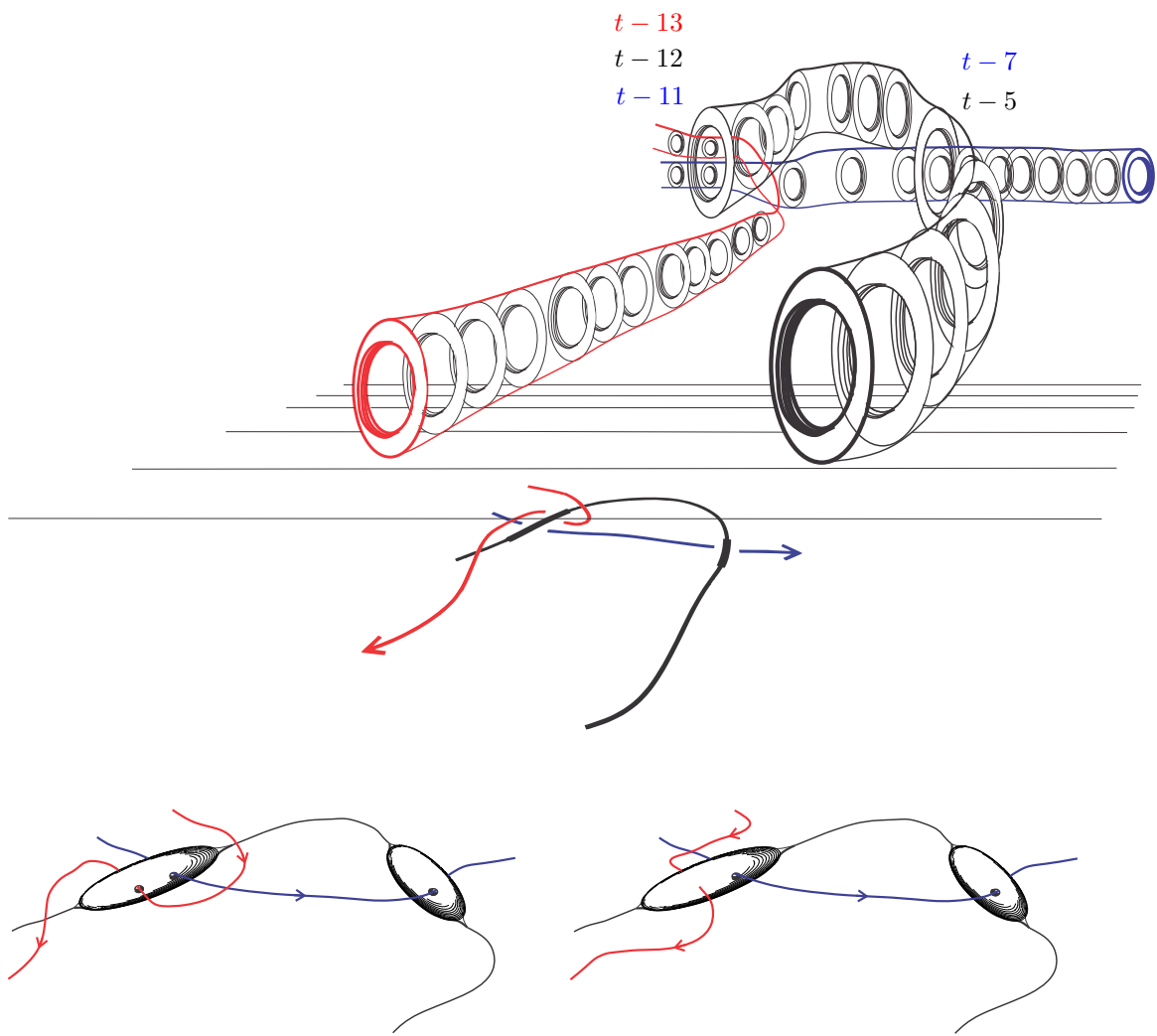}
\caption{\label{F:FlyingRings}``Flying rings'' representation of sphere-and-interval tangles.}
\end{figure}

The above discussion ignores framing, and is therefore applicable to quandle machines. In the rack case, a rack machine represents a projection $K\colon\, M\times[0,1]\to \mathds{R}^4$. Such a figure is called a \emph{framed sphere-and-string tangle}. We think of $K(M\times\{0\})$ as `\emph{the tangle}', and we call it simply `$K$', and of $K(M\times\{1\})$ as `\emph{the framing curve}'.

\subsection{Equivalence of sphere-and-interval tangles}\label{SS:sansequivalence}

We now pass to low dimensional topology, by defining two sphere-and-interval tangles to be \emph{equivalent} if they are related by smooth ambient isotopy.

\begin{defn}[Equivalence and stable equivalence of sphere-and-interval tangles]
Two sphere-and-inverse tangles $T_1$ and $T_2$ are \emph{equivalent} if there exists a smooth homeomorphism $h\colon\, \mathds{R}^4\times [0,1]\to \mathds{R}^4$ with $h(T_1\times \{0\})=T_1$, and $h(T_1\times \{t\})$ a sphere-and-interval tangle for all $t\in[0,1]$, and $h(T_1\times \{1\})=T_2$. If we additionally require stabilization (replacing a segment of an interval by a trivially embedded sphere which doesn't link with anything) and its inverse, then $T_1$ and $T_2$ would be said to be \emph{stably~equivalent}.
\end{defn}

The Reidemeister moves correspond to the compact supported ambient isotopies in Figure~\ref{F:moves4d2}.

\begin{figure}[htb]
%\begin{minipage}{0.8\textwidth}
\psfrag{r}[c]{\small\emph{R1}}\psfrag{s}[c]{\small\emph{R2}}\psfrag{t}[c]{\small\emph{R3}}
\centering
\includegraphics[width=0.95\textwidth]{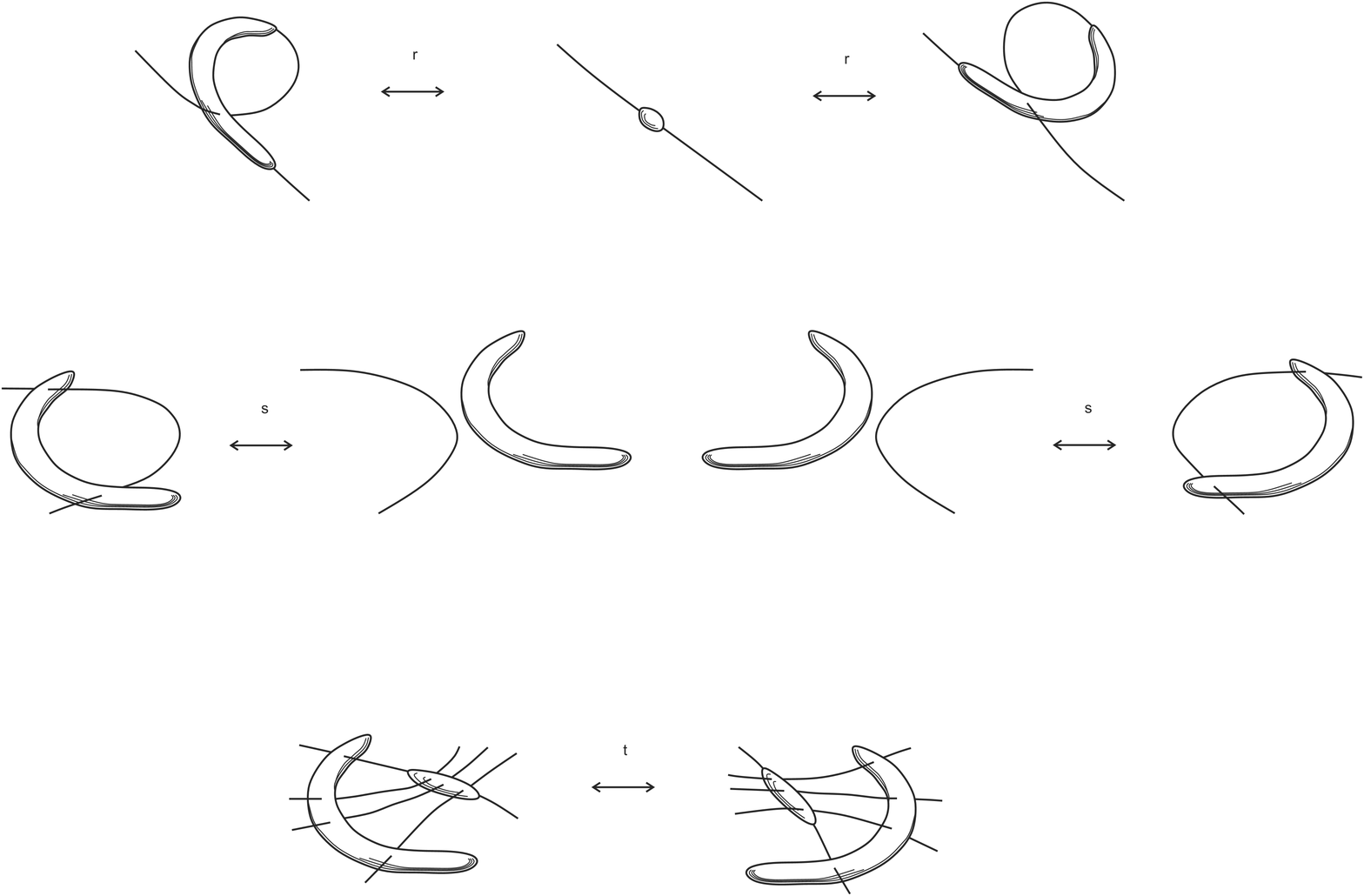}
%\end{minipage}
\caption{\label{F:moves4d2} Ambient isotopies on sphere-and-interval tangles corresponding to Reidemeister moves on their diagrams.}
\end{figure}

\begin{thm}[Reidemeister Theorem for Machines]\label{T:MachineReidemeisterTheorem}
Two machines are (stably) equivalent if and only if any two sphere-and-interval tangles which they represent are (stably) equivalent.
\end{thm}

\begin{proof}
One direction of the theorem is clear--- each Reidemeister move realizes an ambient isotopy of a sphere-and-interval tangle.

To prove the converse, embed the sphere-and-interval tangle $T$ into a larger space which we call $\mathcal{X}$. An elements of $\mathcal{X}$ is a collection of jointly embedded $3$--balls $B_1,B_2,\ldots,B_k$ together with an embedded collection of closed intervals $I_1,I_2,\ldots, I_l$, each of which at its endpoints meets the boundaries of balls, and does not meet the balls anywhere else. We may also allow a collection of embedded rays, each of which has endpoint on the boundary of a ball and is otherwise disjoint from the rest of the picture. Consider $\accentset{\circ}{X}\ass X\setminus \bigcup_{i=1}^k\textrm{Int}(B_i)$, which is an embedded object in $\mathds{R}^4$. The explicit parametrization of a sphere in a sphere-and-interval diagram exhibits it as the boundary of a ball, and this we may think of a sphere-and-interval tangle as an $\accentset{\circ}{X}$ for an element $X$ of $\mathcal{X}$.

To obtain a sphere-and-interval tangle from an element $X\in\mathcal{X}$, work one ball $B$ at a time, with respect to a projection $\pi$ to a fixed but generic $3$--dimensional hyperplane $H$. For ease of exposition, pretend that $B$ is a cube $[x_1,x_2]\times[y_1,y_2]\times[z_1,z_2]$, and that $\pi (B)$ intersects other balls and intervals only at $\{x_1,x_2\}\times[y_1,y_2]\times[z_1,z_2]$. Such sloppiness is standard in $4$--dimensional topology--- \citep{Kirby:89} famously begins with the words ``\ldots the phrase ``corners can be smoothed'' has been a phrase that I have heard for 30 years, and this is not the place to explain it''). To control embedded elements inside $\pi (B)$, choose a stratified Morse function $f$ for $\accentset{\circ}{I}\cap B$ (\citep{GoreskyMacPherson:88}. By compactness, $\pi B$ contains images of a finite number of critical points of $f$. Inside a small neighbourhood, each critical point is of one of the forms below:

 \begin{equation}
\centering
\includegraphics[width=0.9\textwidth]{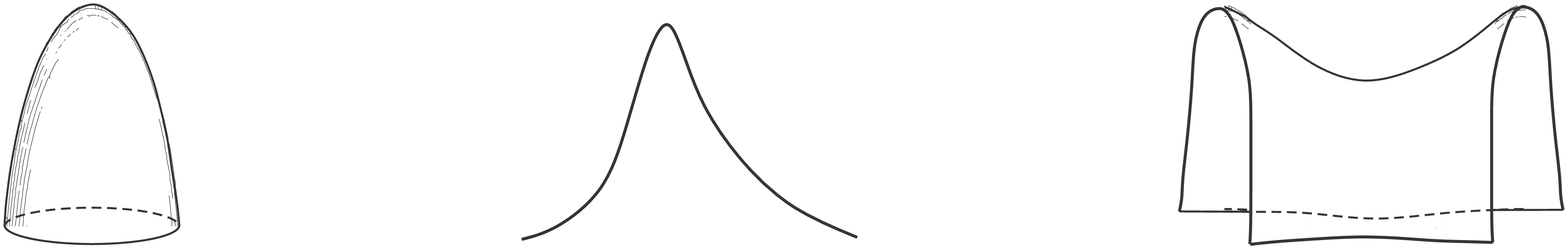}
\end{equation}

Choose a point $p\in B$ with $x$--coordinate $x\in[x_1,x_2]$. For sufficiently small $\epsilon>0$ there are no critical points of $f$ in $B^1\ass [x-\epsilon,x+\epsilon]\times[y_1,y_2]\times[z_1,z_2]\subseteq B$. As we shrink $[x_1,x_2]$ to $[x-\epsilon,x+\epsilon]$, the boundary of $\pi (B)$ will cross over critical points of the image of $f$. By induction and by general position, after shrinkage this ball contains only line segments between the planes $\{x-\epsilon\}\times[y_1,y_2]\times[z_1,z_2]$ and $\{x+\epsilon\}\times[y_1,y_2]\times[z_1,z_2]$ without critical points, and also $2$--dimensional components (parts of boundaries of other balls) without critical points. Next, cut out $\pi (B^1)$, scale it to a ball $B^2$ of radius $\epsilon$ around $p$, and connect endpoints and end-lines on $\pi (\partial B^1)$ to endpoints and end-lines on $\pi (\partial B^2)$ with straight lines and broken surfaces without critical points. For sufficiently small epsilon, there will be no $2$--dimensional components intersecting $\partial B^1$. The embedded element of $\mathcal{X}$ which we obtain is independent of the order by which we shrink the balls. Up to reparametrization this is a sphere-and-interval tangle.

It remains to prove that different choices of the point $p$ lead to sphere-and-interval diagrams which differ by Reidemeister moves, and that smooth ambient isotopy of an element $\mathcal{X}$ changes the resulting sphere-and-interval tangle by Reidemeister moves. The first fact is essentially a special case of the second, so we prove only the second.

Generically choose a $3$--dimensional hyperplane $H\subset\mathds{R}^4$, on which we draw a Roseman diagram of an element $X\in\mathcal{X}$. If this is a Roseman diagram of sphere-and-interval tangle, then we already know how to rewrite that diagram as a Reidemeister diagram. We would like to know how to do this for a general element of $X\in\mathcal{X}$. In order to piggy-back on the results of \citep{Carter:12}, embed each $1$-stratum (interval) in $X$ as a curve on the boundary of a cylinder:

 \begin{equation}
\centering
\includegraphics[width=0.5\textwidth]{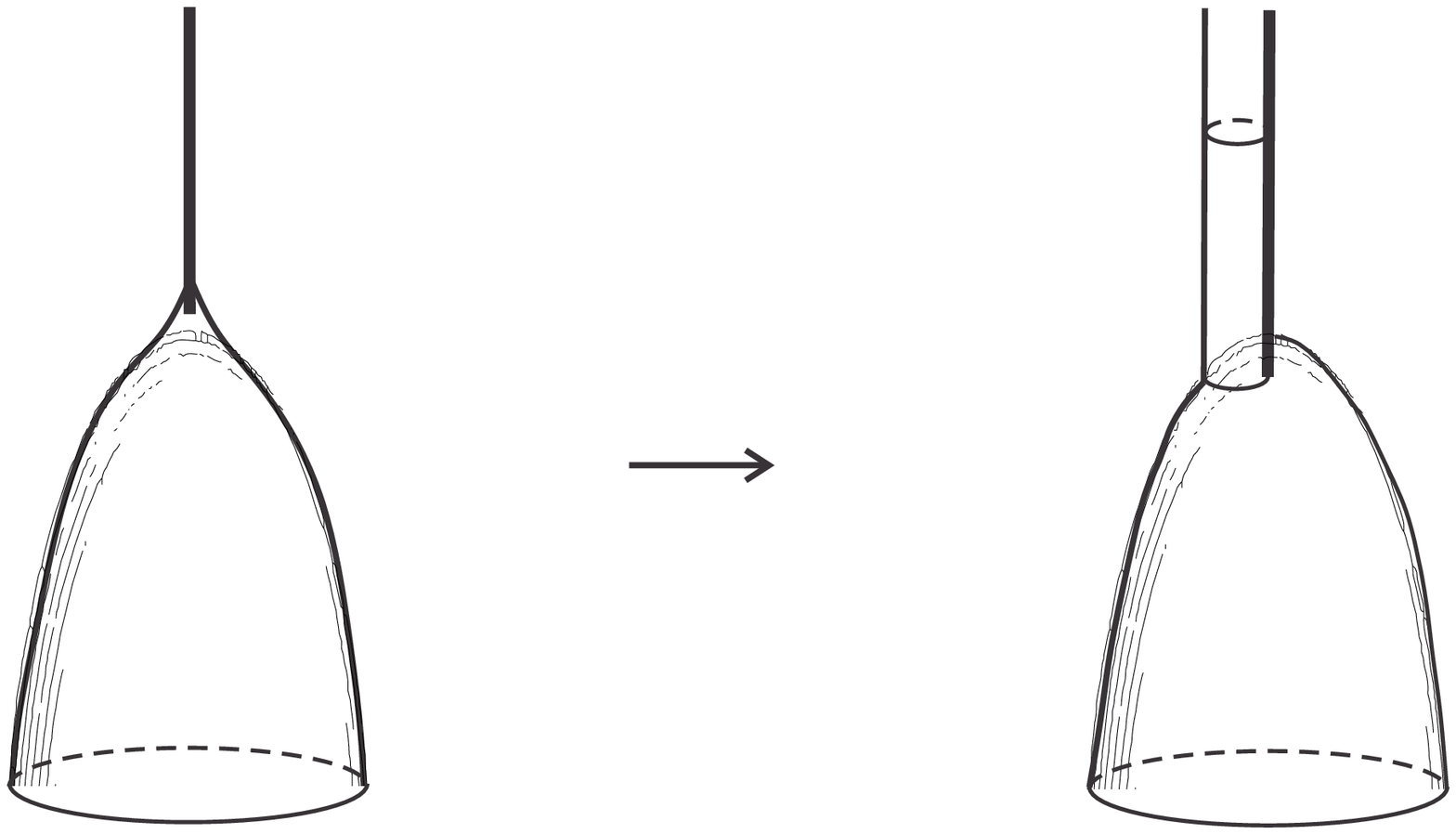}
\end{equation}

The spheres and cylinders now together constitute a \emph{foam}, and Carter proves that any two Roseman diagrams in $H$ of a foam representing $\accentset{\circ}{X}$ corresponding to different Morse functions are related by a finite sequence of local moves in which an isolated critical point of a Morse function is pushed through a plane in the diagram. Deleting the cylinder (which was a cosmetic construction of convenience) and leaving only the 1-stratum in its boundary reduces this collection of moves to those shown in Figure~\ref{F:Roseman}.

 \begin{figure}[htb]
\centering
\includegraphics[width=\textwidth]{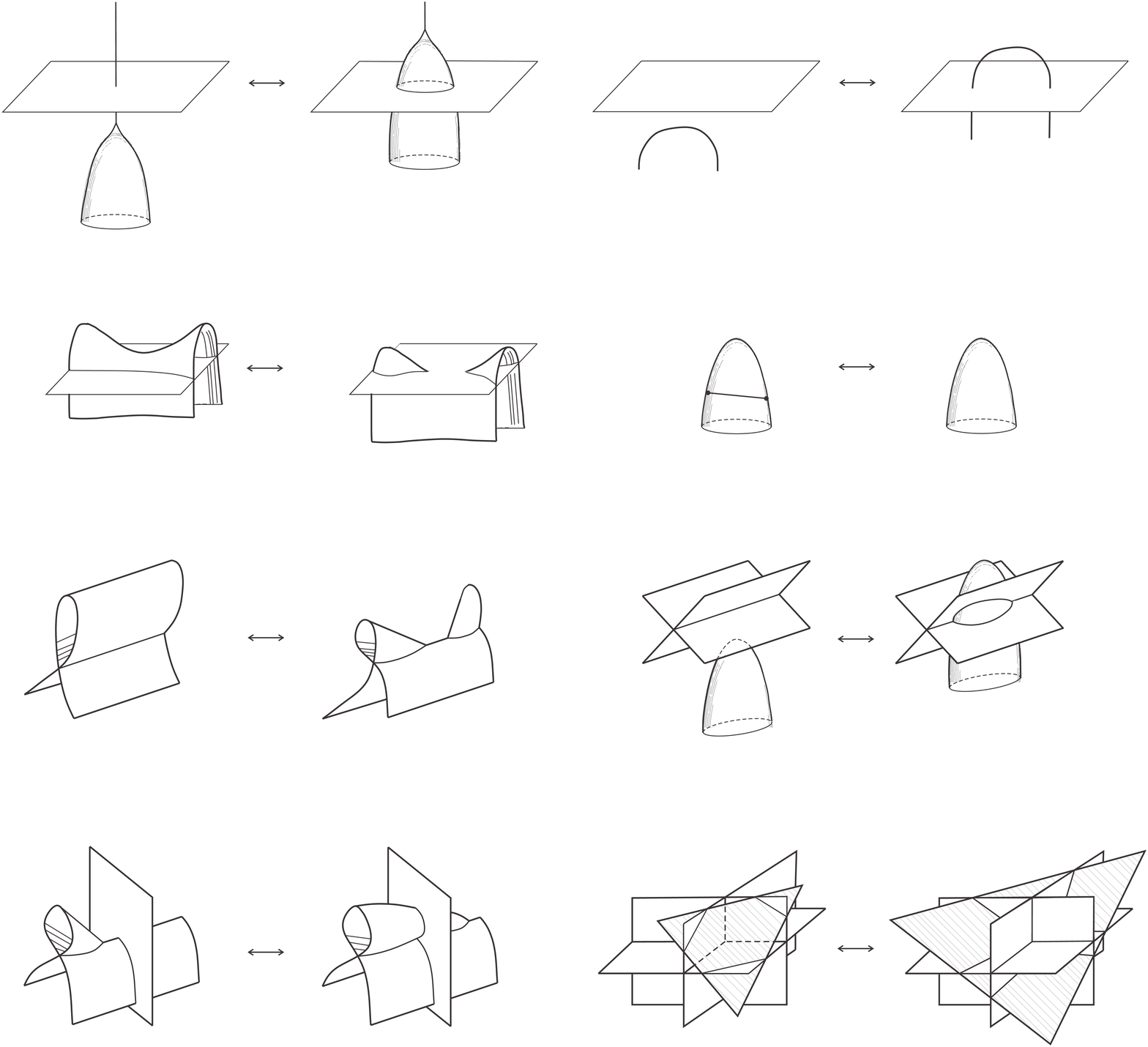}
\caption{\label{F:Roseman} Roseman moves for sphere-and-interval tangles.}
\end{figure}

Inside a $3$--ball, if the critical point is of a $1$--dimensional stratum and if $x\in[x_1,x_2]$ lies below it, then the local move results in a sphere-and-interval tangle whose Reidemeister diagram differs from the original by an R2~move.

 \begin{equation}
\centering
\includegraphics[width=0.6\textwidth]{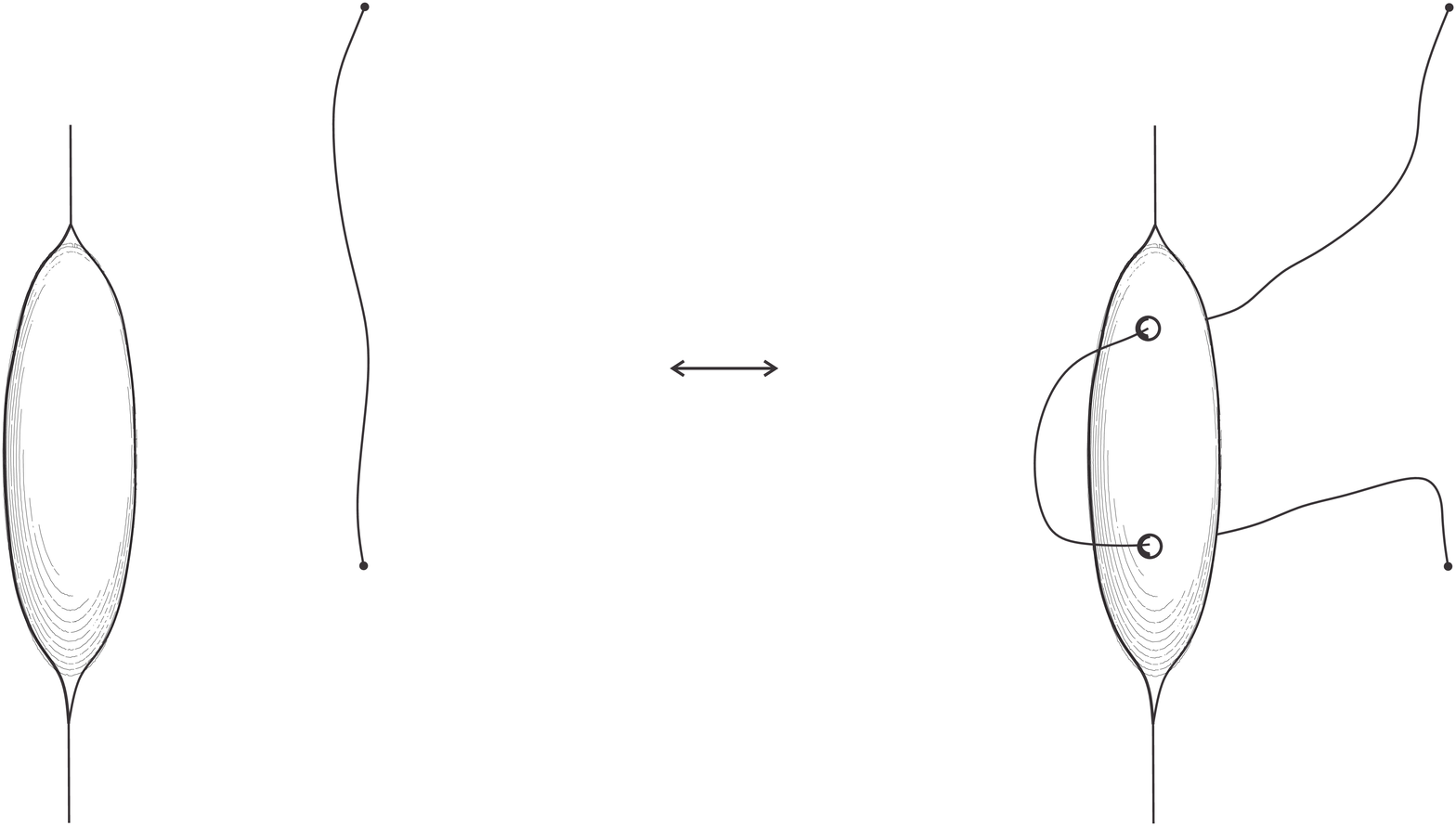}
\end{equation}

If the critical point is of a $2$--dimensional stratum and if $x\in[x_1,x_2]$ lies below it, then the local move results in a sphere-and-interval tangle whose Reidemeister diagram differs from the original by an R3~move.

\begin{equation}\label{E:PushMove}
\psfrag{r}[r]{\small\emph{R3}}
\centering
\includegraphics[width=0.8\textwidth]{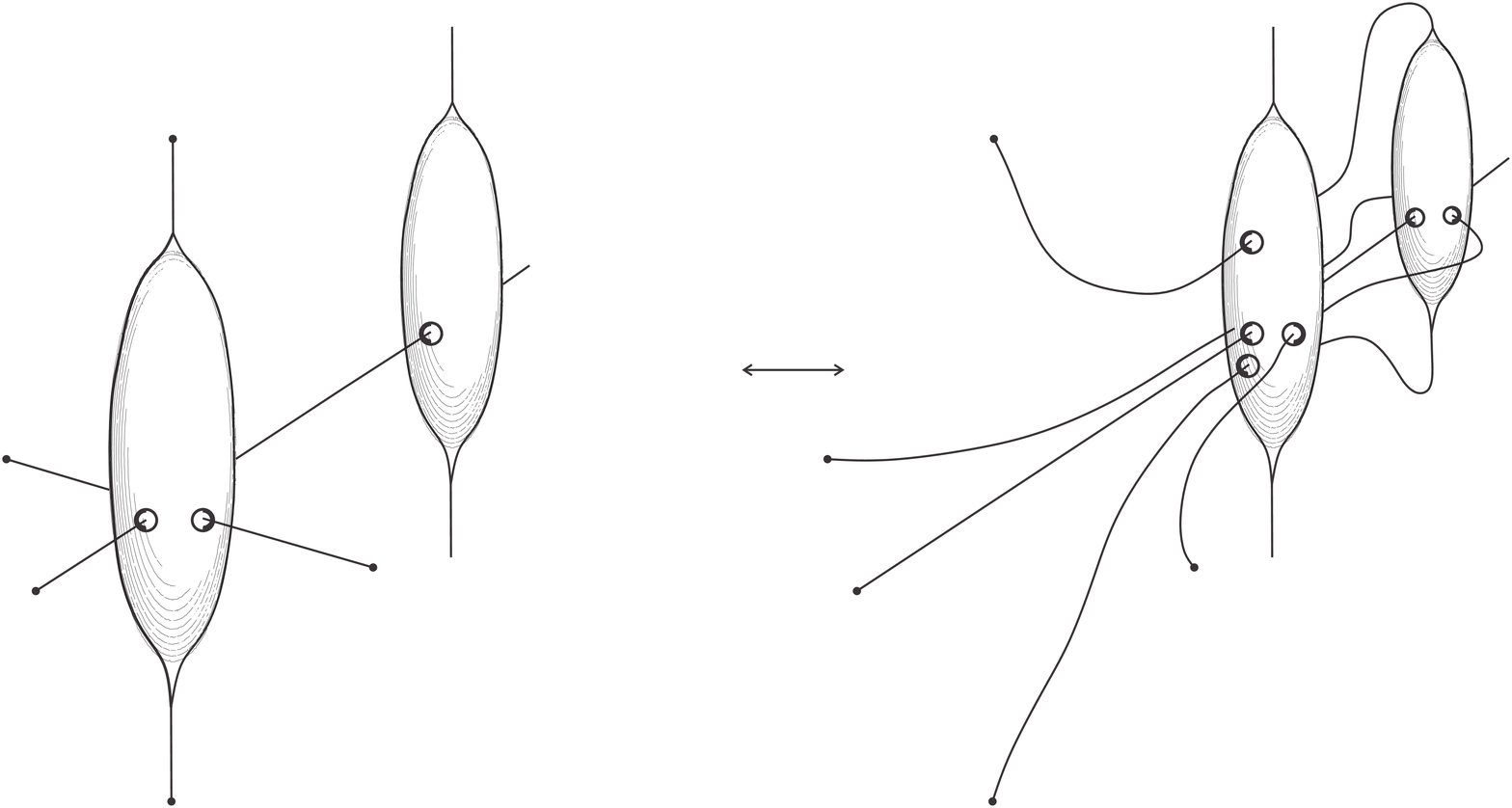}
\end{equation}

Finally, twisting (possible only in the unframed case) corresponds to an R1~move.

 \begin{equation}
\centering
\includegraphics[width=0.5\textwidth]{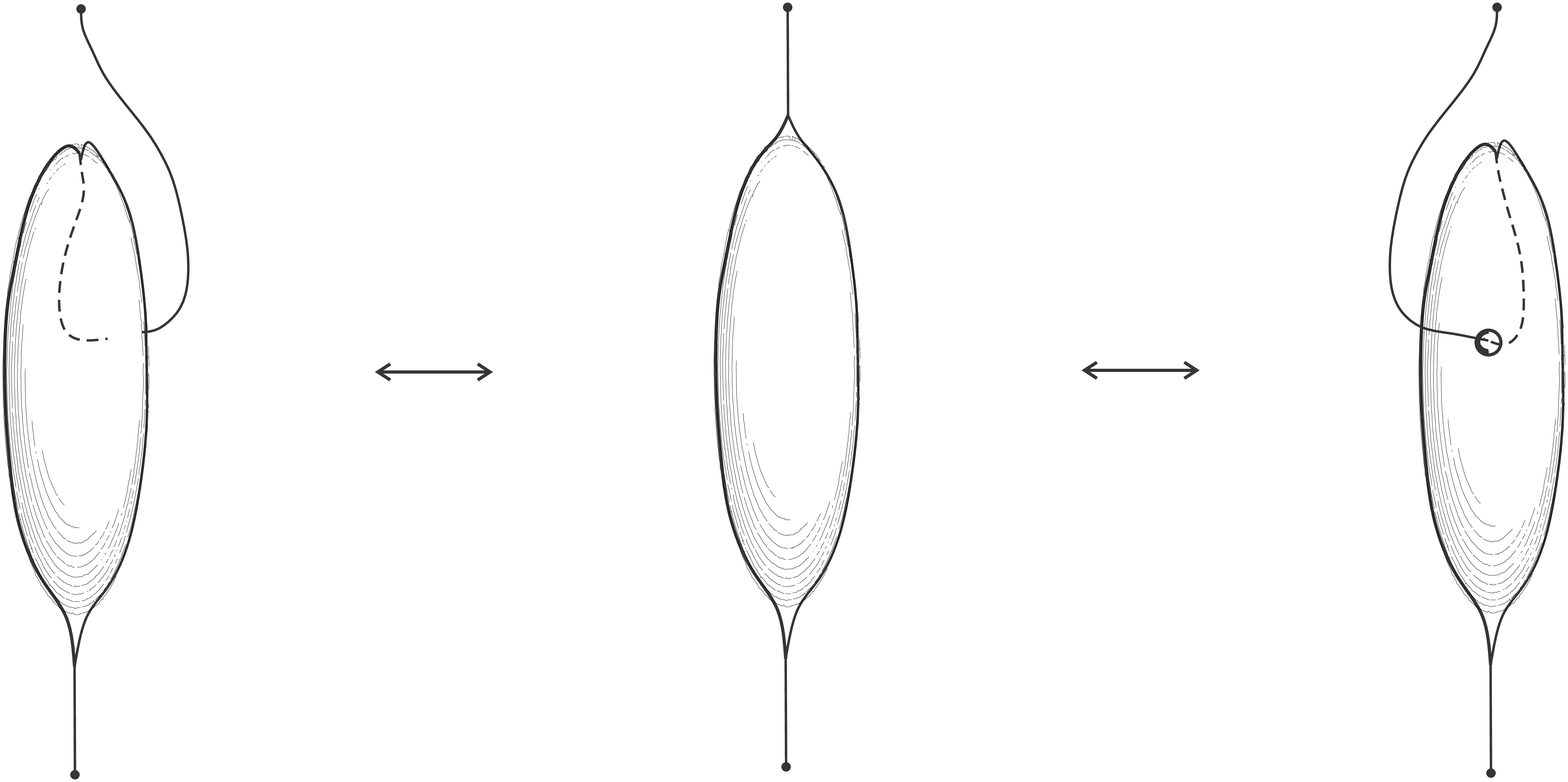}
\end{equation}

Finally, creation of over-strands corresponds to inflation of spheres.
\end{proof}

\begin{rem}
A more direct proof than to embed a sphere-and-interval tangle in a foam would have been to imitate Carter's and Roseman's arguments using stratified Morse theory. The non-manifold points in a sphere-and-interval tangle are isolated and are zero-dimensional, so the minor modifications to their proof which we would require are straightforward.
\end{rem}

\begin{rem}
Following Roseman \citep{Roseman:98}, we prove our Reidemeister Theorem in the smooth~category instead of in the piecewise-linear (PL) category in which Reidemeister proved his result \citep{Reidemeister:32}, because we don't know a combinatorial set of moves which generate PL ambient isotopy in the way that `triangle moves' generate ambient isotopy in dimension $3$ \citep{Graeub:50}.
\end{rem}

\subsection{Rosemeister diagrams}\label{SS:Rosemeister}

A Reidemeister diagram is a planar diagram of an essentially non-planar object. As a result, it contains `virtual' crossings with no topological meaning, and we must take into account a whole slew of virtual and semi-virtual moves, in addition to the usual Reidemeister moves, in order to account for these. A Roseman diagram is too general, as the Roseman moves are irrelevant for the study of machines, as the various singular points in and between spheres all disappear when we project down to a machine. In addition, a Roseman diagram cannot be coloured. We therefore propose a compromise, which we call a \emph{Rosemeister diagram}, which the authors think may be the best diagrammatic representation for tangle machines of all.

\begin{figure}[htb]
\centering
\includegraphics[width=4.5in]{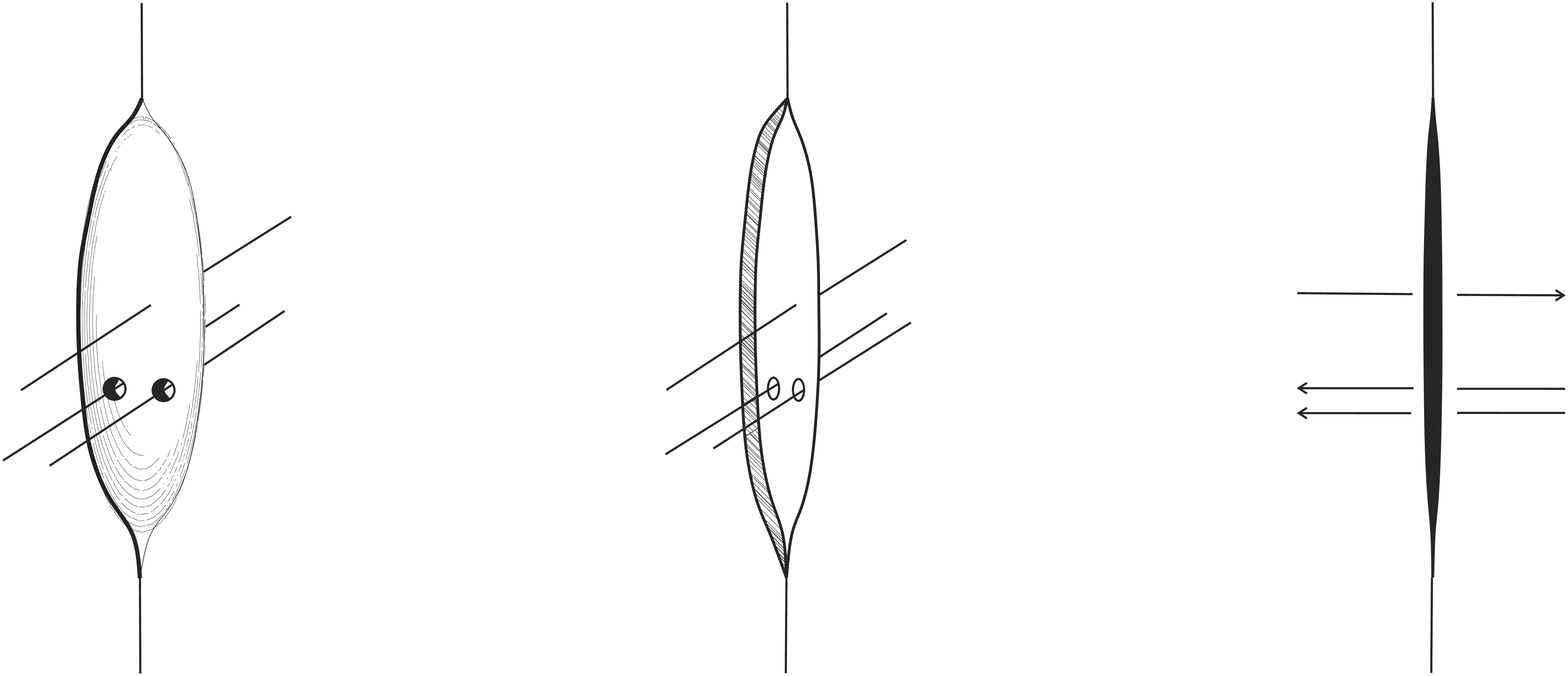}
\caption{\label{F:DimReduce}A Roseman diagram, a Rosemeister diagram, and a Reidemeister diagram of a single interaction.}
\end{figure}

In a Rosemeister diagram for a sphere-and-interval tangle $T$, we crush the spheres in Roseman diagrams for $T$ to discs. By eliminating their interiors, we do away with Roseman moves which we don't need, while keeping the advantage of a Roseman diagram, that is not requiring virtual crossings. Note that interval segments can pass right through interval segments in Rosemeister diagrams, just as in Roseman diagrams.

Interval sections can be coloured by elements of the rack $Q$ which colours the machine, with the colour changing as we pass through the discs, which inherit the colours of the registers which they represent. We could not have done this for Roseman diagrams, as we would not have known how to colour interval segments as they pass through the interior of spheres.

Reidemeister \textrm{I} for Rosemeister diagrams is related to R1 for Roseman diagrams by ambient isotopy. See Figure~\ref{F:R1Rose}.

\begin{figure}[htb]
\centering
\includegraphics[width=0.7\textwidth]{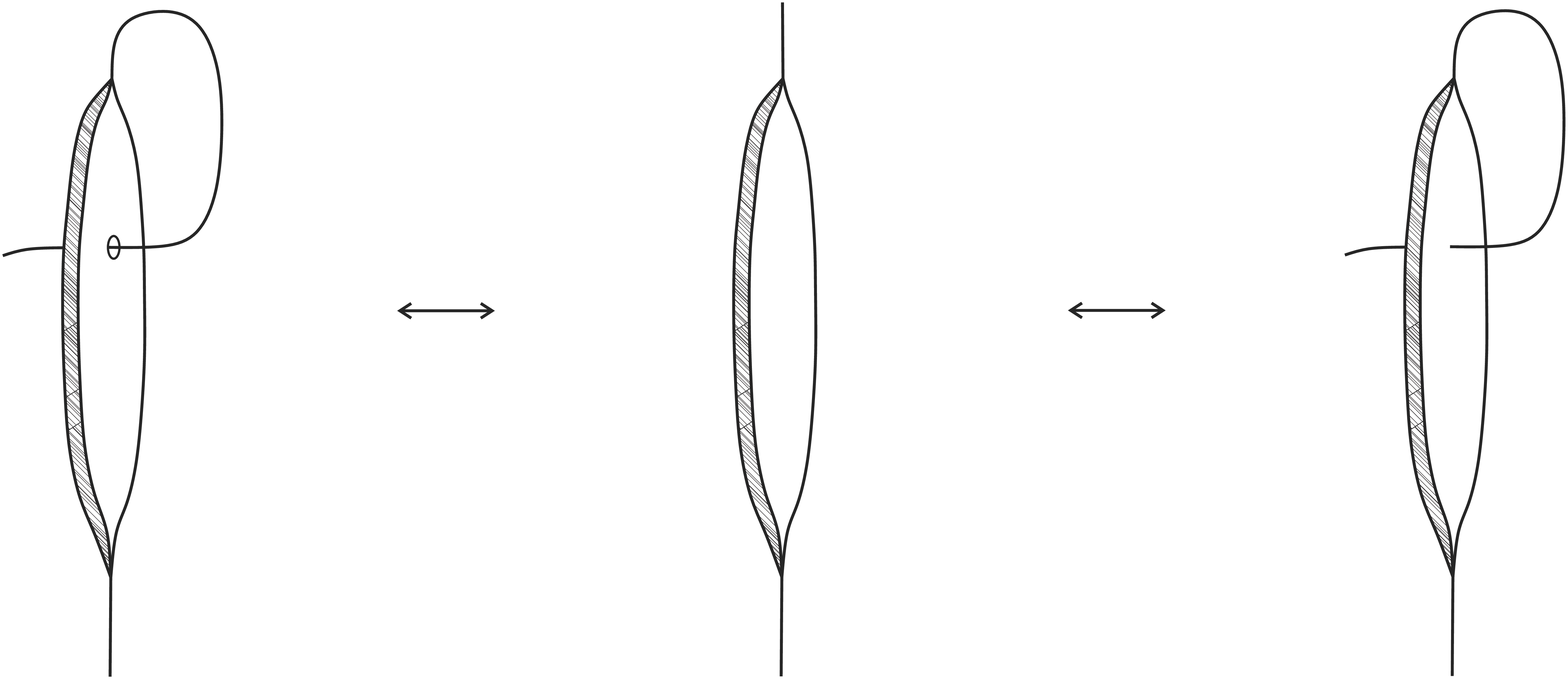}
\caption{\label{F:R1Rose} An R1 move for a Rosemeister diagram.}
\end{figure}

\subsection{Relationship with ribbon torus knots}\label{SSS:wknotsrel}

In the appendix to the previous paper, we discussed w-knots and ribbon torus knots. The diagrammatic calculus of w-knots is similar to the diagrammatic calculus of tangle machines, and indeed cutting up w-knots into \emph{w-knotted tangles} has been represented by a \emph{ball and hoop model} which is similar to our sphere-and-interval tangle \citep{BarNatan:13}. %There is also a notion of \emph{w-knotted graphs}. See \textit{e.g.} \citep{BarNatanDancso:13}.

There is no well-defined map from a w-tangle to a sphere-and-interval tangle or vice versa. However, the space of \emph{equivalence classes} of w-tangles is a quotient of the space of \emph{stable equivalence classes} of tangle machines by false stabilization (Equation~\ref{E:FalseStabilization}).

\begin{thm}
The space of equivalence classes of w-tangles is isomorphic to the quotient of the space of stable equivalence classes of tangle machines by false stabilization.
\end{thm}

We explain the above result. A \emph{w-tangle} is an algebraic object obtained as a concatenation of {\includegraphics[width=20pt]{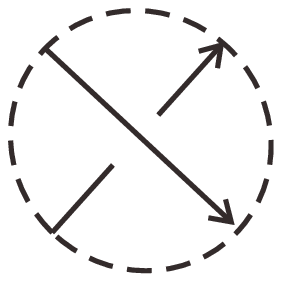}} and {\includegraphics[width=20pt]{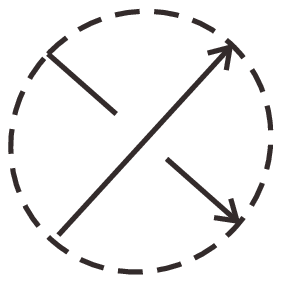}} in the plane.
Two w-tangles are \emph{equivalent} if they are related by a finite sequence of Reidemeister moves as shown in Figure~\ref{F:WTangleReidemeister}.
Thus, the difference between equivalence classes of w-tangles and of diagrams tangle machines lies in the over-strands. True and false stabilization combine to suppress over-strands, so that Reidemeister moves for tangle machines coincide, in the quotient, with Reidemeister moves for w-tangles.

\begin{figure}[htb]
\psfrag{T}[r]{\small \emph{VR1}}
\psfrag{R}[r]{\small \emph{VR2}}
\psfrag{S}[r]{\small \emph{VR3}}
\psfrag{Q}[r]{\small \emph{SV}}
\psfrag{A}[r]{\small \emph{R2}}
\psfrag{B}[r]{\small \emph{R3}}
\psfrag{C}[r]{\small \emph{UC}}
\psfrag{D}[r]{\small \emph{R1}}
\centering
\includegraphics[width=0.9\textwidth]{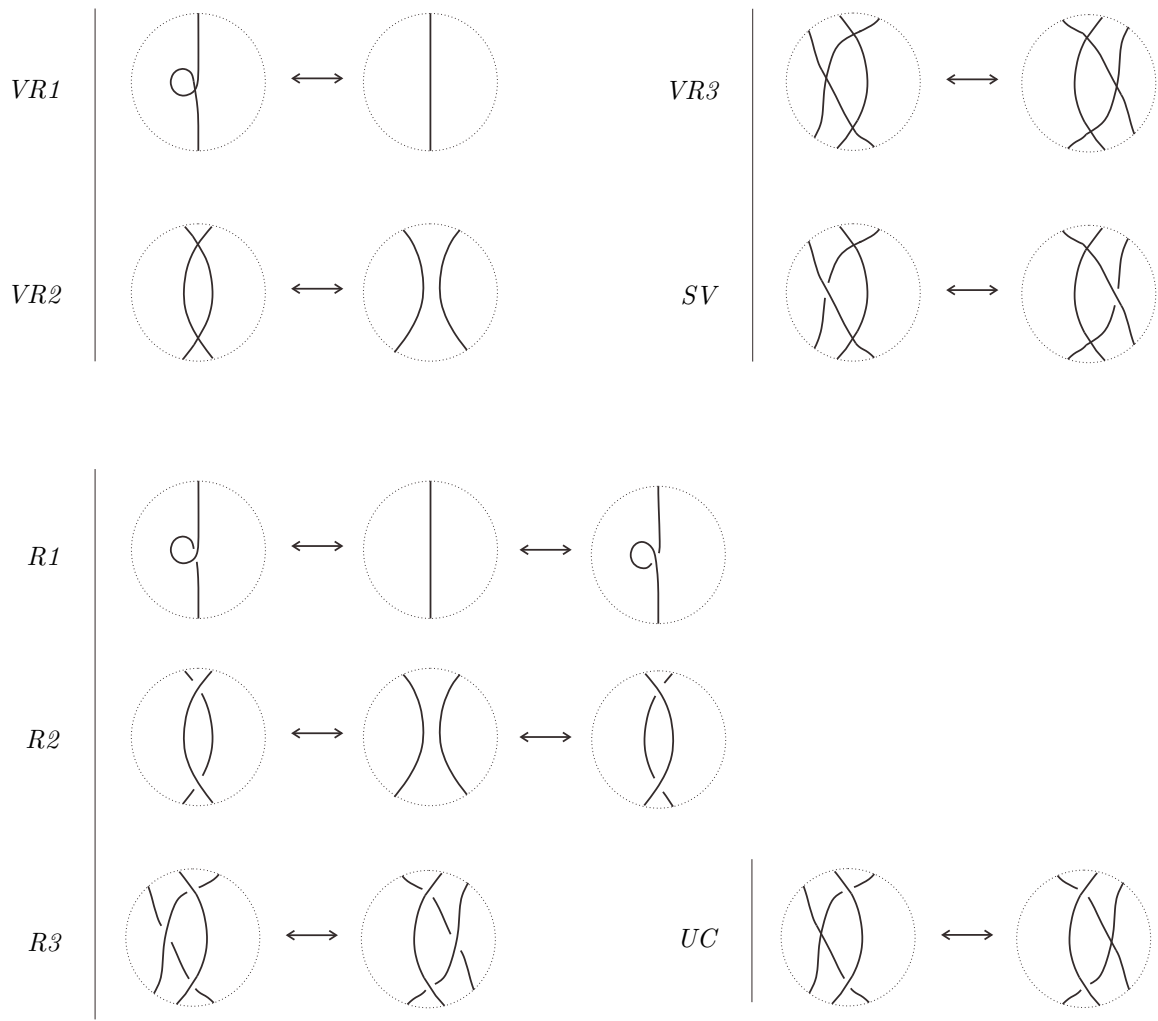}
\caption{\label{F:WTangleReidemeister}Reidemeister moves for w-tangles.}
\end{figure}

\begin{defn}\label{D:UnderlyingWKnot}
If a w-tangle $K$ corresponds to a stable equivalence class of tangle machines to which our machine $M$ belongs, then we say that $K$ is the \emph{underlying w-tangle of $M$}.
\end{defn}

Note that w-knotted objects also admit Rosemeister diagrams, as illustrated in Figure~\ref{F:WKnot}.

\section{Some elementary invariants}
\label{SS:SimpleInvariants}

In this section we describe some simple characteristic quantities associated to equivalence classes of machines. Such quantities are called \emph{invariants}.

\begin{rem}
Category theory allows a precise definition: Invariants are functors out of a category of tangle machines whose morphisms are equivalences, or out of a closely related category.
\end{rem}

\begin{defn}
An \emph{information invariant} is an invariant that is additive under connect sums:
\begin{equation}
v(M_1\begin{minipage}{8pt}\includegraphics[width=8pt]{hash.eps}\end{minipage} M_2)=v(M_1)+v(M_2).
\end{equation}
\end{defn}

\begin{rem}
An example of an invariant that is not an information invariant would be `stick number'--- the minimal number of straight line segments with which a diagram of an equivalent machine can be constructed.
\end{rem}

We are interested only in information invariants; we consider that the information content of a machine is the collection of its information invariants. The identity functor is also an information invariant, so there is a sense in which an equivalence class of machines, or a `best' representative inside it, is the information content of a machine. Our characterization of information content therefore serves to focus our attention on invariants valued in more familiar categories such as categories of numbers, polynomials, etc.

An invariant is called \emph{stable} if it is an invariant of stable equivalence classes.

\subsection{Underlying graph and reduced graph}\label{SSS:UnderlyingGraph}

The graph $G$ which underlies $M$ is unchanged by Reidemeister moves, and is thus an invariant of machine equivalence classes (a \emph{machine invariant}). It is not stable.

\begin{defn}
The \emph{reduced graph} $\widetilde{G}$ of $M$ is the graph obtained from the underlying graph $G$ of $M$ by contracting all $2$--valent vertices of $G$ with one of their incident edges.
\end{defn}

 The reduced graph is a stable invariant. For a tangle machine for example, the reduced graph will be a collection of isolated vertices and loops, and will count the number of open and of closed processes in the machine.

 The underlying graph is unaffected by connect sum. It tells us how many interactions a machine has, how many registers are contained in each, and whether they are open or closed.

%The \emph{order} of a machine is its number of registers, and its \emph{size} is its number of edges. Both are invariants of equivalence classes of machines, but neither are stable.

\subsection{Initial and terminal colour sets}\label{SSS:InitialColourSet}

\begin{defn}
Let $r_1,\ldots,r_\nu$ and $s_{1},\ldots,s_{\nu}$ denote the initial and the terminal registers of the open processes $P_1,\ldots,P_\nu$ of machine $M$, correspondingly. The set $\{\rho(r_1),\ldots,\rho(r_\nu)\}$ is called the \emph{initial colour set} of $M$, and $\{\rho(s_1),\ldots,\rho(s_\nu)\}$ is called the \emph{terminal colour set} of $M$.
\end{defn}

We record the following observation, which was used in the examples in the prequel to this paper.

\begin{prop}
Initial and terminal colour sets of a machine, indexed so that $s_i$ is the terminal register for the process whose initial register is $s_i$ for all $i=1,2,\ldots,\nu$, are stable machine invariants.
\end{prop}

The initial and terminal colour sets are also unaffected by connect sum. They provide some measure as to the computation that a machine is carrying out. If they are very different from one another, then that implies that the machine must have at least a certain number of interactions.

%The transition from initial to terminal colour sets represents the computation of the machine.

\subsection{Nontrivial interaction number and nonunit interaction number}\label{SSS:InteractionNumber}

\begin{defn}[Trivial interaction]\label{D:TrivialInteraction}
An interaction in a machine $M$ with agent register $r$ is \emph{trivial} if $M$ is equivalent to a machine $M^\prime$ in which $r$ has no patients. An interaction in $M$ with agent register $r$ is \emph{unit} if $M$ is equivalent to a machine $M^\prime$ in which all patients of $r$ share the same colour $x\in R$ as $r$.
\end{defn}

\begin{defn}[Nonunit interaction number, Nontrivial interaction number]\label{D:NonTrivInterNum}\hfill
The number of nontrivial (nonunit) interactions in a machine is called the \emph{nontrivial (nonunit) interaction number} of the machine.
\end{defn}

Stabilization and Reidemeister moves to not add or take away nontrivial and nonunit interactions, thus both the nontrivial interaction number and the nonunit interaction number of a machine are stable invariants.

\begin{rem}\label{R:TrivialUndecidable}
 Triviality of an interaction is undecidable. For example, let $Q$ be the quandle whose elements are elements of a group $G$ with undecidable word problem (see \textit{e.g.} \citep{Miller:92}) and whose operation is $g\trr h \ass h^{-1}gh$. Then it is undecidable in $Q$ whether the colour of an agent $r$ is equal to the colour of an input $r^\prime$.
\end{rem}

The nontrivial interaction number and the nonunit interaction number are both additive under connect sums.

\subsection{Fundamental rack}\label{SSS:UniversalRack}%underlying w-knot

In this section, the word ``rack'' should be changed to ``quandle'' when we are discussing quandle machines.

Given a machine $M$, we can discard the colouring $\rho$, and instead colour the registers by distinct formal symbols $c_1,c_2,\ldots,c_N$, subject to the axioms of a rack, and subject to the rule that the output corresponding to input $x$ and operator $y$ is $x\trr y$. Thus, for example,

\begin{equation}
\begin{tikzcd}[row sep=1em,
column sep=1em]
x \rar[dash] & \ooplus \rar & y \\
& z \arrow[dash, dashed]{u} &
\end{tikzcd}
\end{equation}

\noindent means that $y=x\trr z$.

The machine, without its colours, thus gives rise to a rack with generators $c_1,c_2,\ldots,c_N$ and with relations dictated by how interactions concatenate inside the machine. This rack is called the \emph{fundamental rack} of $M$ and is denoted $Q(M)$. It is a stable machine invariant. Because no relations have been introduced beyond those forced on us by the machine $M$ itself, the colouring $\rho$ must factor through a rack homomorphism $Q_M\to Q$.

The fundamental rack
%has a topological realization, described in Section \ref{SSS:RealizeQuandle}. It
is unchanged by contraction of an edge that is not in the image of $\bm{\phi}$, and so it descends to an invariant of the w-knotted graph underlying the machine, described in Definition~\ref{D:UnderlyingWKnot}.

\begin{rem}
There is a notion of a \emph{birack}, which is a more powerful notion than a rack in which colours change at undercrossings and also at overcrossings \citep{FennJordanKauffman:04}. Biracks give rise to invariants of w-knots \citep{BartholomewFenn:11} and therefore also to invariants of machines.
\end{rem}

The fundamental rack of a connect sum is the free product of fundamental racks of summands. Because any colouring of the machine must factor through the fundamental rack, the fundamental rack represents the maximum amount of data that a machine can contain.

\subsection{Linking graph}\label{SSS:LinkingGraph}

The `linking' of a machine $M$ with underlying graph $G$ and with processes $P_1,P_2,\ldots,P_\nu$ is captured as follows:

\begin{defn}[Linking number; linking vector; (unframed) linking graph]\label{D:LinkingGraph}
 The \emph{linking number} of register $r$ with process $j$ is the number of edges $e$ in process $j$ such that $\phi(e)=r$ and $\mathrm{sgn}(e)=+$, minus the number of edges $e$ in process $j$ such that $\phi(e)=r$ and $\mathrm{sgn}(e)=-$. The \emph{linking graph} $\mathrm{Link}(M)$ of $M$ is a labeling of each vertex in $G$ by a \emph{linking vector} $v^r\ass \left(v^r_{1},v^r_2\ldots,v^r_{\nu}\right)$ whose $k$th entry is the linking number of $r$ with process $k$. The \emph{unframed linking graph} $\mathrm{Link}_0(M)$ is the labeled graph obtained by setting to zero the entry in each linking vector $v^r$ which represents the interactions of $r$ with its own process $P$.
 \end{defn}

 \begin{rem}
 The notion of a machine's linking graph parallels the notion of the linking matrix of a classical link, as in \textit{e.g.} \citep{Kauffman:01}.
 \end{rem}

The linking graph is an invariant of a rack~machine, and the unframed linking graph is an invariant of a quandle~machine. This is because an R2 move cancels or creates a pair of inverse interactions $\trr$ and $\rrt$ by the same agent, while an R3 move has no effect on any linking vector, and the effect of an R1 move is only on the `diagonal' entries.

The linking graph of a connect sum is obtained by adding linking vectors at each vertex. It is also a measure of the complexity of a machine. To illustrate, consider the following example:

\begin{example}
Let $M$ be a machine coloured by the Alexander quandle, that is the quandle whose elements are rational functions in a real variable $t$ and whose operation is $x\trr y= (1-t)x+ty$. The degree of $x\trr y$ is $1+\max(\mathrm{Deg}(x),\mathrm{Deg}(y))$. If the linking vectors are all zero and the machine is connected, it implies that the degrees of all colours in the machine share the same degree. Indeed, the gap between the highest and the lowest degrees of Alexander quandle colours which can appear in each component of a machine is completely determined by the linking graph.
\end{example}

\begin{defn}[Reduced linking graph]\label{D:ReducedLinkingGraph}
The \emph{reduced linking graph} $\widetilde{\mathrm{Link}}(M)$ of a linking graph $\mathrm{Link}(M)$ is the labeled graph obtained from $\mathrm{Link}(M)$ by first deleting all zero entries in all linking vectors in $M$, and then by removing all $2$--valent vertices with empty linking vector from the graph (contracting an edge incident to them). The  \emph{reduced unframed linking graph} $\widetilde{\mathrm{Link}_0}(M)$ is defined analogously.
\end{defn}

The reduced linking graph is a stable invariant of a rack~machine, and the graph obtained by setting all `diagonal entries' to null is a stable invariant of a quandle~machine. The reduced linking graph is a more compact way than the linking graph of expressing the same complexity information.

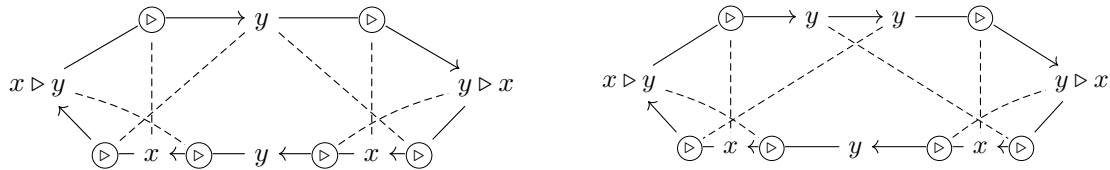
\begin{figure}[htb]
\centering
\begin{minipage}[t]{0.45\linewidth}
\vspace{0pt}
\resizebox{.99\hsize}{!}{\begin{tikzcd}[ampersand replacement=\&, row sep=1em, column sep=0.6em]
\& \& \ooplus  \arrow[dash, dashed]{dd} \arrow{rrr} \& \& \& y \arrow[dash,dashed]{ddllll} \arrow[dash,dashed]{ddrrrr} \arrow[dash]{rrr} \& \& \&  \ooplus \arrow[dash,dashed]{dd} \arrow{drr} \& \& \\
x\trr y \arrow[dash]{urr} \arrow[dash, dashed, bend left=10]{drrr} \& \& \& \& \& \& \& \& \& \& y\trr x  \arrow[dash]{dl} \\
\&  \ooplus \arrow{ul} \rar[dash] \& x \&  \ooplus \lar \& \& y \arrow[dash]{ll} \&  \&  \ooplus \arrow[dash,dashed, bend left=10]{urrr} \arrow{ll}  \& x \arrow[dash]{l} \&  \ooplus \lar \&
\end{tikzcd}}

\end{minipage}
\hspace{10pt}
\begin{minipage}[t]{0.45\linewidth}
\vspace{0pt}
\resizebox{.99\hsize}{!}{\begin{tikzcd}[ampersand replacement=\&,row sep=1em, column sep=0.45em]
\& \& \ooplus  \arrow[dash, dashed]{dd} \arrow{rr} \& \& y \arrow{rr} \arrow[dash,dashed]{ddrrrrr}\& \& y \arrow[dash]{rr} \arrow[dash,dashed]{ddlllll} \& \&  \ooplus \arrow[dash,dashed]{dd} \arrow{drr} \& \& \\
x\trr y \arrow[dash]{urr} \arrow[dash, dashed, bend left=10]{drrr} \& \& \& \& \& \& \& \& \& \& y\trr x  \arrow[dash]{dl} \\
\&  \ooplus \arrow{ul} \rar[dash] \& x \&  \ooplus \lar \& \& y \arrow[dash]{ll} \&  \&  \ooplus \arrow[dash,dashed, bend left=10]{urrr} \arrow{ll}  \& x \arrow[dash]{l} \&  \ooplus \lar \&
\end{tikzcd}}
\end{minipage}
\caption{The right-hand side machine has one linking vector all of whose entries are $1$, whereas one of the entries in the linking vector for the single process in the left-hand machine is $2$. The rack used in both machines satisfies $x \trr y = y \rrt x$.}
\label{fig:div}
\end{figure}

An entry in the linking vector indicates the total influence of an individual register on the various processes in a machine. When we do not need all of the information in the linking graph, a marginalized version may be useful:

\begin{defn}[Linking matrix]\label{D:LinkingMatrix}
Let $P_1,P_2,\ldots,P_\nu$ denote the processed of a machine $M$. Denote registers of the $i$th process $P_i$ in a machine $M$ by $r_i^1,r_i^2,\ldots,r_i^k$, whose respective linking vectors are
 \begin{equation}v(r_i^j)\ass \left(v_{1}(r_{i}^j),v_2(r_i^j)\ldots,v_{\nu}(r_i^j)\right) \qquad \text{for $j=1,2,\ldots,k$.}\end{equation}
Set $R_i \in \mathbb{R}^{\nu}$ to be the $1\times\nu$ vector whose $j$th entry is $\sum_{s=1}^k \abs{v_j(r_i^s)}$, that is  the sum taken over all registers in $P_i$ of the absolute values of their respective linking vectors. The \emph{linking matrix} of $M$ is the $\nu\times\nu$--matrix whose rows are $R_1, \ldots, R_\nu$.
\end{defn}

\begin{example}
Consider the following (two-process) machine of which the $j$th register in the $i$th process is labeled $x_{ij}$.
\begin{equation}
\begin{minipage}[t]{0.28\linewidth}
\vspace{15pt}
\psfrag{a}[c]{\small $x_{21}$}
\psfrag{b}[c]{\small $x_{24}$}
\psfrag{c}[c]{\small $x_{11}$}
\psfrag{d}[r]{\small $x_{13}$}
\includegraphics[width=0.93\textwidth]{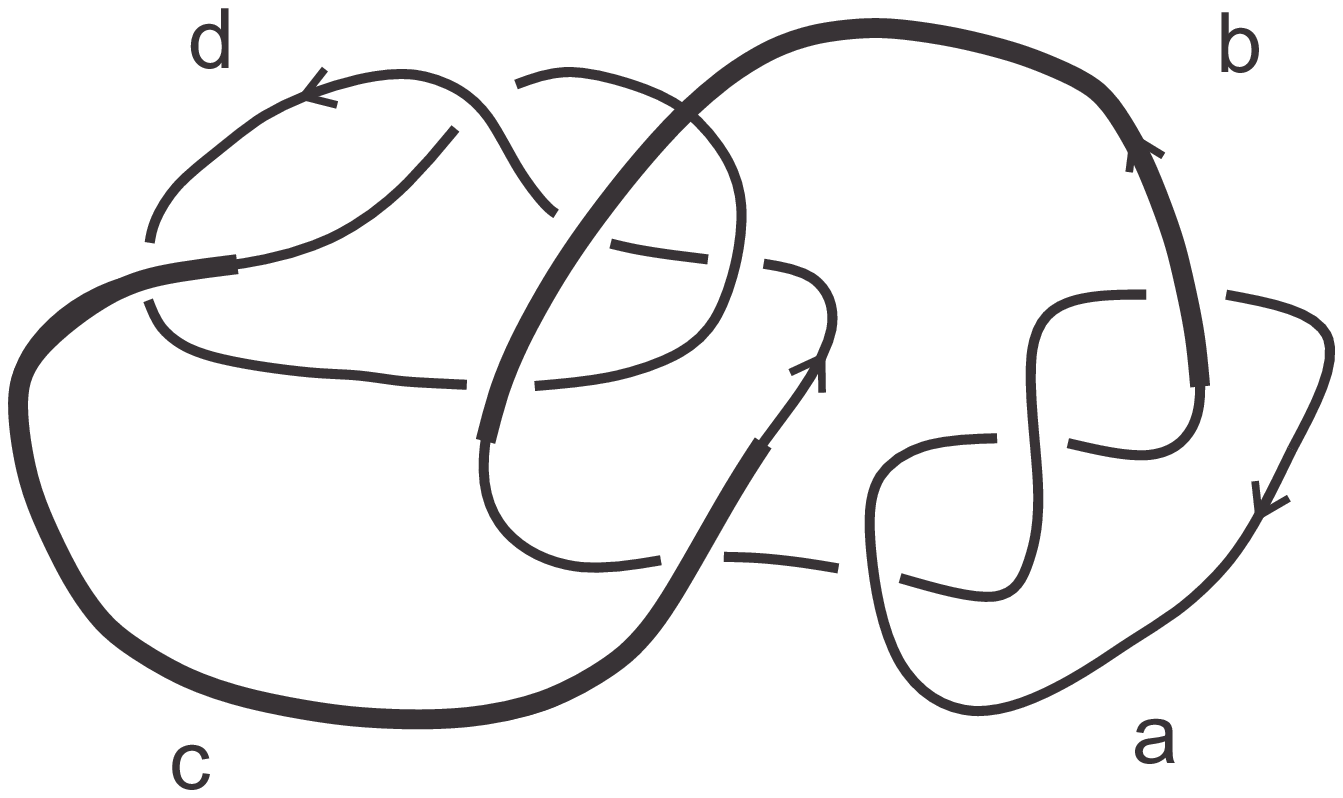}
\end{minipage}
\begin{minipage}[t]{0.6\linewidth}
\vspace{0pt}
\begin{tikzcd}[row sep=0.7em, column sep=0.36em]
& x_{13} \arrow[dash, dashed, bend left=15]{dddrr} \arrow[dash]{ld} & \ooplus \arrow{l} \arrow[dash, dashed, bend left=20]{rrr} & x_{12} \arrow[dash]{l} &  & : & & \ooplus \rar & x_{22} \arrow[dash]{r} \arrow[dash, dashed, bend left=10]{ddr} & \ooplus \arrow{dr} \\
\oominus \arrow{d} & & & & \oominus \arrow{ul}    & & x_{23} \arrow[dash]{ur} & & & & x_{21} \arrow[dash, dashed, bend left=10]{lllu} \arrow[dash]{ld} \\
x_{14} \arrow[dash]{dr} & & & & x_{11} \arrow[dash, dashed, bend left=10]{llllu} \arrow[dash]{u} \arrow[dash, dashed, bend right=30]{rrr}  & & & \ooplus \arrow{ul} & \arrow[dash]{l} x_{24} \arrow[dash, dashed, bend left=10]{uur} \arrow[dash, dashed, bend left=30]{llllld} \arrow[dash, dashed, bend left=70]{llluu} \arrow[dash, dashed, bend left=50]{dlllllll} & \ooplus \arrow{l} \\
& \oominus \rar & x_{15} \arrow[dash, dashed, bend left=10]{rruu} \arrow[dash]{r} & \oominus \arrow{ur} &
\end{tikzcd}
\end{minipage}
\end{equation}
For this machine the linking graph, $\mathrm{Link}(M)$, and its corresponding stabilization, $\mathrm{Link}_0(M)$, (depicted below using squiggly arrows)
are obtained as
\begin{equation}
\begin{tikzcd}[row sep=0.8em, column sep=0.25em]
\small
& (-1,0) \arrow{ldd} \arrow[squiggly, bend right=10]{dddr} & & (0,0) \arrow{ll} &  & & & & (0,1) \arrow{rrd} \arrow[squiggly, bend left=10]{rrd} & \\
& & & &     & & (0,0) \arrow{urr} & & & & (0,1) = v(x_{21}) \arrow{lld} \arrow[squiggly, bend left=10]{lld} \\
(0,0) \arrow{drr} & & & & (-1,1) = v(x_{11}) \arrow{uul} \arrow[squiggly, bend right=10]{uulll}  & & & & \arrow{llu} (0,1) \arrow[squiggly, bend left=10]{uu} & \\
& & (-1,0) \arrow{urr}  \arrow[squiggly, bend right=20]{urr} & &
\end{tikzcd}
\end{equation}
Their unframed counterparts are
\begin{equation}
\begin{tikzcd}[row sep=0.8em, column sep=0.25em]
\small
& (0,0) \arrow{ldd} & & (0,0) \arrow{ll} &  & & & & (0,0) \arrow{rrd} & \\
& & & &     & & (0,0) \arrow{urr} & & & & (0,0) = v(x_{21}) \arrow{lld} \\
(0,0) \arrow{drr} & & & & (0,1) = v(x_{11}) \arrow{uul} \arrow[loop left, squiggly]{} & & & & \arrow{llu} (0,0) & \\
& & (0,0) \arrow{urr} & &
\end{tikzcd}
\end{equation}
and the framed and unframed linking matrices are, respectively, $\begin{bmatrix} 3 & 1 \\ 0 & 3 \end{bmatrix}$, and $\begin{bmatrix} 0 & 1 \\ 0 & 0 \end{bmatrix}$.
\end{example}

%\begin{rem}If $M=M_1\Hash M_2$ then $\mathrm{Link}(M)=\mathrm{Link}(M_1)+\mathrm{Link}(M_2)$.\end{rem}

\subsection{Colour linking graph}\label{SSS:SpectralGraph}

Complementary to the notion of a linking graph, which makes no account of colours, there is the notion of the colour linking graph which sees `only' the colours at each register, and so which somehow measures `linking of colours'. Once again, let $M$ be a machine coloured by a rack $Q$ with underlying graph $G$.

\begin{defn}[Colour linking spaces and graphs; unframed and reduced versions]\label{D:SpectralLinking}
Let $e=(v,w)$ be an edge in $G$ with $\phi(e)=r$. If $\mathrm{sgn}(e)=+$ then let $\hat{e}$ denote the automorphism $\trr \rho(v,e)$ of $Q$, and if $\mathrm{sgn}(e)=-$ then let $\hat{e}$ denote the automorphism $\rrt \rho(w,e)$ of $Q$. Denote the space of \emph{inner automorphisms} of $Q$, that is automorphisms of the form $\trr x$ for $x\in Q$, by $\mathrm{Inn}(Q)$. Denote its abelianization, that is its quotient by elements of the form $\trr x\trr y \rrt x\rrt y$, by $\mathrm{Ab}\left(\mathrm{Inn}(Q)\right)$.

The \emph{colour linking space} of a register $r$ is the set $\mathrm{Spec}(r)$ of all proper maximal ideals of the subspace of $\mathrm{Ab}\left(\mathrm{Inn}(Q)\right)$ that is generated by $\hat{e}$ for all $e\in\phi^{-1}(r)$. The \emph{unframed colour linking space} is obtained by ignoring all contributions of half-edges in the same process as $r$. The \emph{(unframed) colour linking graph} of $M$ is a labeling of each vertex of $G$ by its (unframed) colour linking space. The \emph{reduced (unframed) colour linking graph} is defined by removing by contraction all $2$--valent vertices of the (unframed) colour linking graph which are labeled by empty spaces.
\end{defn}

The colour linking space can be informally thought of as the dimension of the set of input colours on which $r$ is acting, modulo the relation implied by R2. It is an analogue of the spectrum of a ring in algebraic geometry.

The (unframed) colour linking graph is unchanged by Reidemeister moves, and so it is an invariant of rack~machines (of quandle~machines). Its reduced version is a stable invariant.

\begin{example}
Consider a machine $M$ with a single interaction, with $10$ edges $e_i\ass (v_i,w_i)$ all of which satisfy $\rho(v_i,e_i)=q$, $\phi(e_i)=r$ and $\mathrm{sgn}(e_i)=+$ for $i=1,2,\ldots,10$ (so $M$ has 21 registers and 11 open processes). Assume also that $Q$ has more than one generator. Then the linking vector of $r$ is $10$, and its colour linking space is generated by $\trr q$. Now, leaving everything else the same, change $\mathrm{sgn}(e_{10})$ to $-$. The linking vector of $r$ becomes $8$, but its colour linking space now vanishes.
\end{example}

\begin{equation}
\begin{tikzcd}[row sep= 3em, column sep=0.8em]
x_1 \rar[dash]& \oodiamond \rar & x_1\diamond y & x_2 \rar[dash]& \oodiamond \rar & x_2\diamond y &\cdots & x_{10} \rar[dash]& \oodiamond \rar & x_{10}\diamond y\\
\ & \ & \ & \ & \ & y \arrow[dash, dashed]{ullll} \arrow[dash, dashed]{ul} \arrow[dash, dashed]{urrr}
\end{tikzcd}
\end{equation}

Because the colour linking space of a register cannot be non-empty in two distinct connect summands, the colour linking space of a register in a direct sum is the union of coloured linking spaces of that register in each of the direct summands. %The colour linking graph is also a measure of complexity. MORE

\subsection{Shannon capacity}\label{SSS:ShannonCapacity}

The intuition behind the following invariant comes from viewing a machine as an information carrier. More formally,
a machine $M$ is a noisy communication channel through which colours as well as interactions are transmitted from A(lice) to B(ob) \citep{Shannon:56}. While $M$ is noisy and non-perfect, the messages on Bob's end appear corrupted and missing. A natural question can then be raised: What is the amount of non-confusable information that can be received by Bob?

Alice has a machine $M$ coloured by a rack $Q$. Alice sends Bob the graph $G$, together with a map equivalent to $\bm{\phi}$, and $k$ values of $\rho$ (not necessarily distinct). For an interaction
\begin{equation}
\begin{tikzcd}[row sep=1em,
column sep=1em]
x \rar[dash] & \ooplus \rar & y \\
& z \arrow[dash, dashed]{u} &
\end{tikzcd}\end{equation}
\noindent we say that any \textbf{pair} of elements of the set $\{x,y,z\}$ can be \emph{confused}. Messages which cannot be confused are called \emph{distinct}. Let $\textrm{Cap}_k(M)$ denote the number of distinct messages of length $k$ which $M$ admits.

\begin{defn}[Shannon capacity]
The \emph{Shannon capacity} of machine $M$ is:
\begin{equation}
\mathrm{Cap}(M)\ass \sup_{k\in\mathds{N}}\sqrt[k]{\mathrm{Cap}_k(M)}
\end{equation}
\end{defn}

\begin{example}
Consider the machine:
\begin{equation}
\begin{tikzcd}[row sep=1em, column sep=1.5em]
& & \ooplus \arrow{drr} \arrow[dash, dashed]{dd} & & \\
x \arrow[dash]{urr} \arrow[dash, dashed, bend left=10]{drrr} & & & & y \arrow[dash]{dl}\arrow[dash, dashed, bend right=10]{dlll}\\
& \ooplus \arrow{ul} \rar[dash] & y\trr x & \ooplus \lar &
\end{tikzcd}
\end{equation}
Any two elements of $Q$ are related by an automorphism, therefore $\mathrm{Cap}_1(M)=1$. A maximal set of distinct messages of length $2$ is $\{xx, xy\}$ and so $\mathrm{Cap}_2(M)=2$. %A maximal set of distinct messages of length $3$ is $\{xxx, xxy, xy(y\trr x)\}$ and so $\mathrm{Cap}_3(M)=3$.
 It seems therefore as though $\mathrm{Cap}(M)=\sqrt{2}$.
\end{example}

The definition of the Shannon capacity of a machine mimics that of the Shannon capacity of a graph \cite{Shannon:56}. It is a stable invariant.

\begin{rem}
A generalization of the above definition would be for Alice to send Bob only partial information about $\bm{\phi}$, and perhaps even no crossing information at all.
\end{rem}

\section{Complexity of machines}\label{S:Complexity}

The goal of this section is to define a complexity measure for a machine paralleling the number of prime factors of a classical knot, link, or tangle, as in Theorem~\ref{T:UniquePrimeFactorization}. The essential feature of our setting is that the colouring plays the lead role, and our definition of coloured prime decomposition may be applied also to classical coloured knots, links, and tangles.

\subsection{Factorization of machines and the definition of complexity}\label{SS:Factorization}

Our view of machines is that they perform computations--- we input colours to some registers, and output resulting colours in other registers. Thus, a machine whose computations are trivial should be considered trivial from the point of view of machine decomposition.

\begin{defn}\label{D:TrivialMachine}
A machine all of whose interactions are trivial (the agent and all patients share the same colour--- see Definition~\ref{D:TrivialInteraction}) is said to be a \emph{unit} machine.
\end{defn}

We define factorization and factors for a machine.

\begin{defn}[Factorization, factors, and prime factors]\label{D:Factor}
If $M$ is equivalent to a connect sum $M_1\begin{minipage}{8pt}\includegraphics[width=8pt]{hash.eps}\end{minipage} M_2$, then $M_1$ and $M_2$ are called \emph{factors} of $M$, and the decomposition of $M$ into $M_1\begin{minipage}{8pt}\includegraphics[width=8pt]{hash.eps}\end{minipage} M_2$ is called a \emph{factorization} of $M$. A machine $M$ is \emph{prime} it is not a unit, and if for any factorization $M=M_1\begin{minipage}{8pt}\includegraphics[width=8pt]{hash.eps}\end{minipage} M_2$, either $M_1$ is a unit or $M_2$ is a unit.
\end{defn}

\begin{defn}[Complexity]\label{D:Complexity}
The \emph{complexity} $\Omega(M)$ of machine $M$ is the maximal $k\in\mathds{N}$ such that $M$ factors into $k$ prime factors.
\end{defn}

\begin{example}
The machine counterpart of the `square knot' are given below. These \emph{square machines} have complexities of, respectively, $1$ and $2$.
\begin{equation}
%\begin{figure}[htb]
%\centering
\begin{minipage}[t]{0.43\linewidth}
\psfrag{x}[c]{\small $y$}
\psfrag{y}[c]{\small $x$}
    \centering
    \vspace{10pt}
    \includegraphics[width=0.9\textwidth]{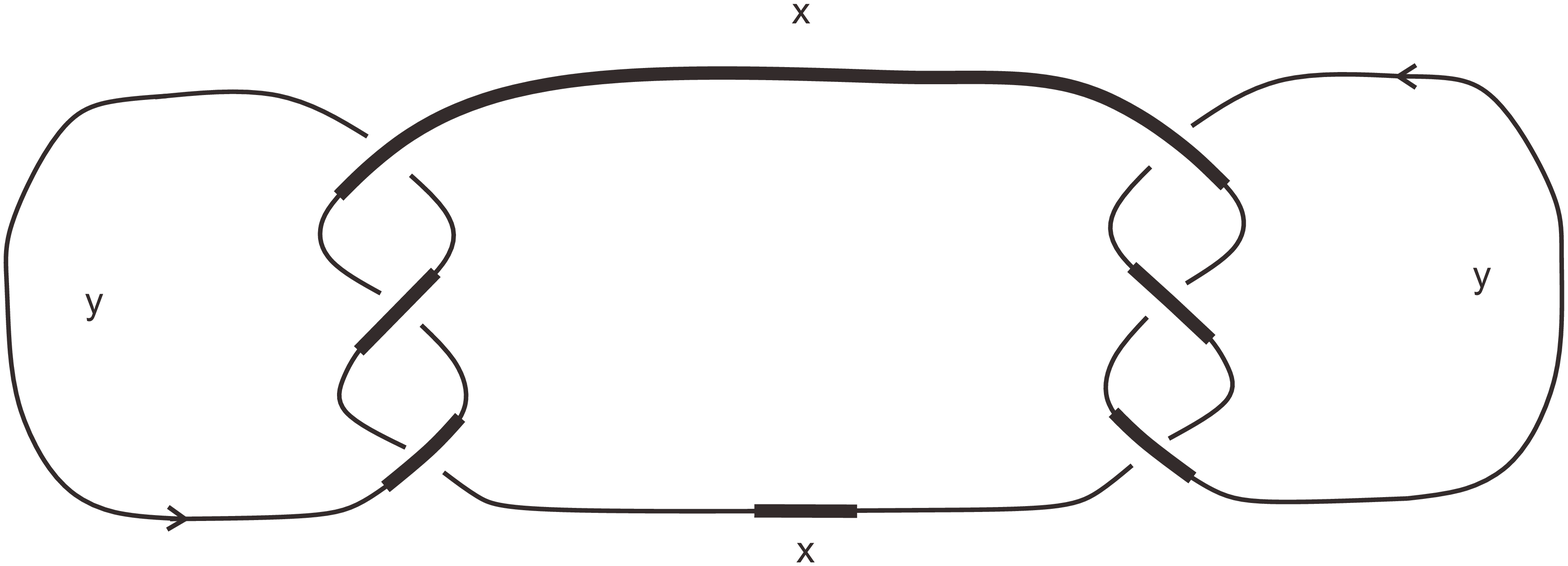}
\end{minipage}
  \hspace{10pt}
\begin{minipage}[t]{0.45\linewidth}
\vspace{0pt}
\resizebox{.99\hsize}{!}{\begin{tikzcd}[ampersand replacement=\&,row sep=1em, column sep=0.7em]
\& \& \oominus  \arrow[dash, dashed]{dd} \arrow{rrr} \& \& \& y \arrow[dash]{rrr} \arrow[dash,dashed]{ddllll} \arrow[dash,dashed]{ddrrrr} \& \& \&  \ooplus \arrow[dash,dashed]{dd} \arrow{drr} \& \& \\
x \arrow[dash]{urr} \arrow[dash, dashed, bend left=10]{drrr} \& \& \& \& \& \& \& \& \& \& x  \arrow[dash]{dl} \\
\&  \oominus \arrow{ul} \rar[dash] \& y \rrt x \&  \oominus \lar \& \& y \arrow[dash]{ll} \& \&  \ooplus \arrow[dash,dashed, bend left=10]{urrr} \arrow{ll}  \& x \trr y \arrow[dash]{l} \&  \ooplus \lar \&
\end{tikzcd}}
\end{minipage} % \\
\end{equation}
%\vspace{20pt}
%
\begin{equation}
\begin{minipage}[t]{0.43\linewidth}
\psfrag{x}[c]{\small $y$}
\psfrag{y}[c]{\small $x$}
    \centering
    \vspace{10pt}
    \includegraphics[width=0.9\textwidth]{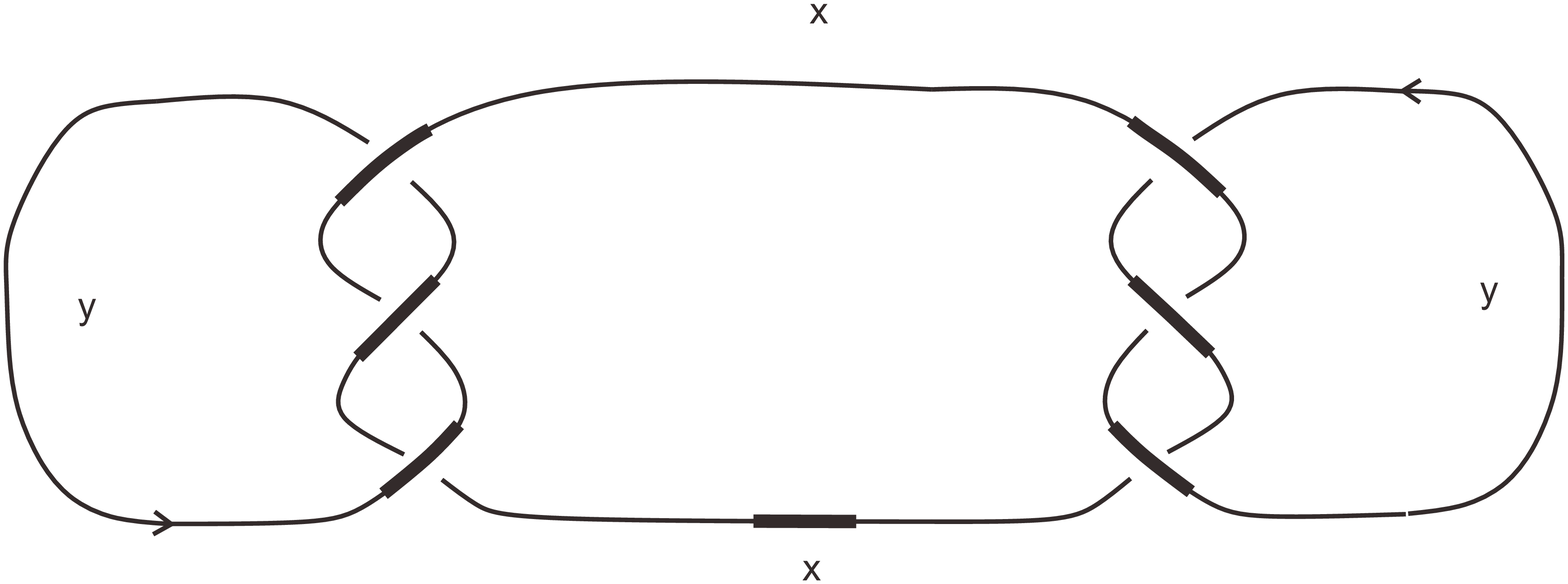}
\end{minipage}
\hspace{10pt}
\begin{minipage}[t]{0.45\linewidth}
\vspace{0pt}
\resizebox{.99\hsize}{!}{\begin{tikzcd}[ampersand replacement=\&, row sep=1em, column sep=0.8em]
\& \& \oominus  \arrow[dash, dashed]{dd} \arrow{r} \& y \arrow{rr} \arrow[dash,dashed]{ddll} \& \& y \arrow[dash,dashed]{ddrr} \arrow[dash]{r} \&  \ooplus \arrow[dash,dashed]{dd} \arrow{drr} \\
x \arrow[dash]{urr} \arrow[dash, dashed, bend left=10]{drrr} \& \& \& \& \& \& \& \& x  \arrow[dash]{dl} \\
\&  \oominus \arrow{ul} \rar[dash] \& y \rrt x \&  \oominus \lar \& y \arrow[dash]{l} \&  \ooplus \arrow[dash,dashed, bend left=10]{urrr} \arrow{l}  \& x \trr y \arrow[dash]{l} \&  \ooplus \lar \&
\end{tikzcd}}
\end{minipage}
\end{equation}

%\caption{Square machines of complexity $1$ (upper) and $2$ (lower).}
%\label{fig:squaregranny}
%\end{figure}

The machine below (which is different from the square machine) has a complexity 2:
\begin{equation}
\begin{minipage}[t]{0.43\linewidth}
\psfrag{x}[c]{\small $y$}
\psfrag{y}[c]{\small $x$}
%    \centering
    \vspace{10pt}
    \includegraphics[width=0.9\textwidth]{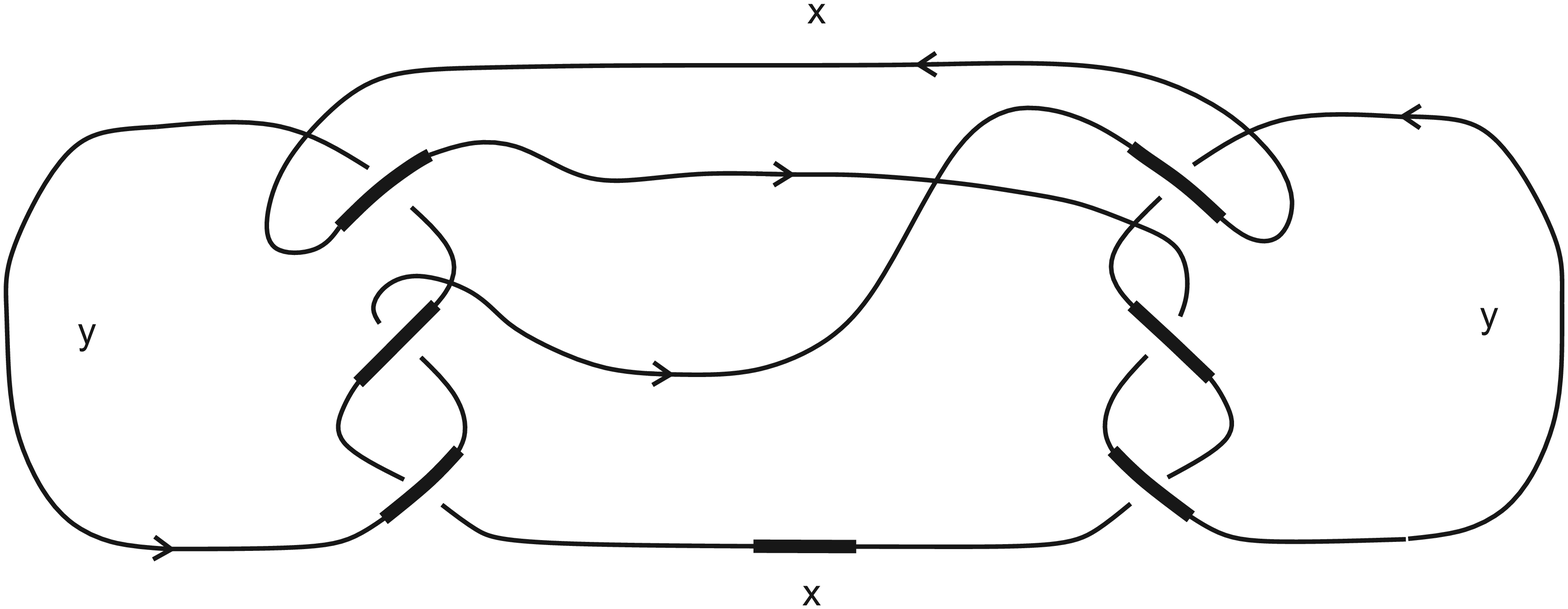}
\end{minipage}
%\hspace{5pt}
\begin{minipage}[t]{0.45\linewidth}
\vspace{0pt}
\resizebox{.99\hsize}{!}{\begin{tikzcd}[ampersand replacement=\&,row sep=1em, column sep=0.8em]
\& \& \oominus  \arrow[dash, dashed]{dd} \arrow{r} \& y \arrow{rr} \arrow[dash,dashed]{ddrrrr} \& \& y \arrow[dash,dashed]{ddllll} \arrow[dash]{r} \&  \ooplus \arrow[dash,dashed]{dd} \arrow{drr} \\
x \arrow[dash]{urr} \arrow[dash, dashed, bend left=10]{drrr} \& \& \& \& \& \& \& \& x  \arrow[dash]{dl} \\
\&  \oominus \arrow{ul} \rar[dash] \& y \rrt x \&  \oominus \lar \& y \arrow[dash]{l} \&  \ooplus \arrow[dash,dashed, bend left=10]{urrr} \arrow{l}  \& x \trr y \arrow[dash]{l} \&  \ooplus \lar \&
\end{tikzcd}}
\end{minipage}
\end{equation}
\end{example}

\begin{thm}[Complexity is an invariant]\label{T:Complexity}
Complexity $\Omega(M)$ is a well-defined stable invariant. It is additive with respect to connect sum.
\end{thm}

\begin{proof}
Complexity is defined as a maximum over an equivalence class, therefore it is a machine invariant. It is a stable invariant, because stabilization preserves all properties in its definition. It is finite because it is bounded above by the number of nonunit interactions in machine $M$. Indeed, for any factorization $\mathcal{N}=N_1\begin{minipage}{8pt}\includegraphics[width=8pt]{hash.eps}\end{minipage} N_2$ of $M$, the number of nonunit interactions in $M$ equals the sum of nonunit interactions in $N_1$ and in $N_2$. If $N_1$ and $N_2$ are non-unit, each must contain at least one nonunit interaction. Domains of interaction functions of connect summands are disjoint, so by its definition as a maximum, complexity is sub-additive with respect to connect sums. Additivity with respect to connect sums follows from unique prime factorization, Theorem~\ref{T:UniquePrimeFactorization}, because a prime factor of a connect summand is also a prime factor of the connect sum.
\end{proof}

\subsection{The effect of false stabilization on complexity}\label{SS:FalseStabilization}

Consider the following move, which is \emph{not} an equivalence although it is a valid modification of a machine.

\begin{defn}[False stabilization]\label{D:FalseStabilization}
The following machine modification is called \emph{false stabilization}.
\begin{equation}%\label{E:NoResolution}%\centernot\longleftrightarrow
\begin{tikzcd}[row sep=1em,
column sep=1em]
\oodiamond & \ & \oodiamond \\
x \arrow{rr} \uar[dash, dashed]  & & x \uar[dash, dashed]
\end{tikzcd}\quad\longleftrightarrow \quad\  \begin{tikzcd}[row sep=1em,
column sep=1em]
\oodiamond  & & \oodiamond\\
\ & x \arrow[dash, dashed]{ul} \arrow[dash, dashed]{ur}
\end{tikzcd} \quad\ \longleftrightarrow \ \ \begin{tikzcd}[row sep=1em,
column sep=1em]
\oodiamond & \ & \oodiamond \\
x \arrow{rr} \arrow[dash, dashed]{urr}  & & x \arrow[dash, dashed]{ull}
\end{tikzcd}
\end{equation}
\end{defn}

In this section we explore the effect on complexity of false stabilization and destabilization, which we call \emph{joining} and \emph{resolution} correspondingly.

\begin{equation}\label{E:Resolution}
\begin{tikzcd}[row sep=1em,
column sep=1em]
\oodiamond & & \oodiamond \\
x \arrow{rr} \uar[dash, dashed]  & & x \uar[dash, dashed]
\end{tikzcd}\,
\begin{tikzcd}
\phantom{\rule{0pt}{12pt}a} \arrow[bend left=40]{rr}{Join} &  & \phantom{b} \arrow[bend left=40]{ll}{Resolve}
\end{tikzcd}
 \,\begin{tikzcd}[row sep=1em,
column sep=1em]
\oodiamond  & & \oodiamond\\
\ & x \arrow[dash, dashed]{ul} \arrow[dash, dashed]{ur}
\end{tikzcd}\,
\begin{tikzcd}
\phantom{a} \arrow[bend left=40]{rr}{Resolve} &  & \phantom{b} \arrow[bend left=40]{ll}{Join}
\end{tikzcd}\,
\begin{tikzcd}[row sep=1em,
column sep=1em]
\oodiamond & \ & \oodiamond \\
x \arrow{rr} \arrow[dash, dashed]{urr}  & & x \arrow[dash, dashed]{ull}
\end{tikzcd}
\end{equation}

\begin{prop}
Joining cannot increase complexity, and resolution cannot decrease complexity.
\end{prop}

\begin{proof}
False stabilization contracts or expands an edge which is outside the domain of $\bm{\phi}$. If registers $r_1$ and $r_2$ join to form register $r$, then perforce $r_1$ and $r_2$ share the same colour. It remains to show that, if both registers incident to $\phi^{-1}(r_1)$ share the same colour as $r_1$ and also both registers incident to $\phi^{-1}(r_2)$ share the same colour as $r_2$, then both registers incident to $\phi^{-1}(r)$ must share the same colour as $r$. This will show that destabilization cannot create new nonunit interactions, and therefore that it cannot increase complexity.

We prove this claim topologically. Let $R_1$ and $R_2$ be standard ball-bounding spheres representing $r_1$ and $r_2$ correspondingly, and let $R$ be a standard ball-bounding sphere representing $r$, all inside a sphere-and-interval tangle for $M$ which by abuse of notation we also denote $M$. Let $P$ be a $2$--dimensional plane intersecting $R$ transversely, so that slicing $R$ along $P$ and smoothing has the effect of separating $R$ into $R_1$ and $R_2$. We also fix a $3$-dimensional hyperplane $H$ with respect to which we draw a Rosemeister diagram $D$ for $M$.

For the duration of this proof, we allow ourselves to act by ambient isotopy on one part of an embedded object while leaving another fixed. Technically this is accomplished by creating a \emph{bicollar} between what moves and what stays fixed, which acts as a `buffer' along which we to interpolate. See \textit{e.g.} \citep{Kosinski:07} for details. We also implicitly smooth all corners, so at every point in our argument, all objects live in the smooth category.
%Briefly and for example, plane $D$ has a neighbourhood $D\times[-\epsilon,\epsilon]$ called a \emph{bicollar} inside the ball $B$ bounded by $K$, and we may shift the half of $K$ obtained by slicing $K$ along $D\times \{-\epsilon\}$, which represents $K_1$, while leaving $B\cap K$ fixed pointwise, using the image of $K\cap (D\times[-\epsilon,\epsilon])$ as a `buffer'. Of course we may also do the same thing for $K_2$. We may thus shift by ambient isotopy a `copy of $K_1$' and a `copy of $K_2$' inside $K$ around $M$ while leaving the rest of the diagram fixed pointwise. \par

Let $I_1$ and $I_2$ be ambient isotopies of $K_1$ and of $K_2$ correspondingly, at the end of which all half-edges passing through $R_1$ share the same colour in $D$ as $R_1$ itself, and the same for $R_2$. Then, by the bicollar argument mentioned above, $I_1\circ I_2$ may be considered as an ambient isotopy of $R$ which leaves the cutting plane $P$ fixed pointwise. At the end of this ambient isotopy, which extends to ball $B$ with boundary $R$, the ball $I_1\circ I_2 (B)$ may intersect other balls bounded by other spheres in the projection to $H$ of the sphere-and-interval tangle for $M$. Imitating the proof of the Reidemeister Theorem for machines, Theorem~\ref{T:MachineReidemeisterTheorem}, we shrink $I_1\circ I_2(R)$ to a small ball around a point $p$ in $P\cap R$ while leaving the rest of $R$ fixed pointwise, interpolating between the original `big ball' and the current `small ball' with line segments and with broken planes, to again obtain a sphere-and-interval tangle for a machine. If all colours of intervals passing through a disc (the projection of a sphere) in $D$ share the same colour as the projected sphere, then they continue to hold the same colour as the projected sphere when the local picture is pushed through another sphere as in \ref{E:PushMove}. Thus we have exhibited a sphere-and-interval tangle for $M$ in which all line segments passing through $R$ share the same colour as $R$ in the Rosemeister diagram $D$, which means that indeed both registers incident to $\phi^{-1}(r)$ share the same colour as $r$.
\end{proof}

\subsection{Unique prime factorization}\label{SS:UniquePrime}

As Figure~\ref{F:m3131} illustrates, the factorization of a machine into prime machines is not unique. But as the same figure illustrates, there are a finite number of such \emph{prime factorizations}. Each of these represents an `equivalence class' of factorizations, and it is unique up to unit factors as a representative of this `equivalence class'. This claim is made precise below.

\begin{figure}[htb]
\centering
\begin{minipage}[t]{1\linewidth}
\psfrag{x}[c]{\small $x$}
\psfrag{y}[c]{\small $y$}
\psfrag{z}[c]{\small $y \trr x$}
\psfrag{s}[c]{\small $x \trr y$}
\psfrag{P}[c]{\small $P_1$}
\psfrag{Q}[c]{\small $P_2$}
\psfrag{R}[c]{\small $P_3$}
\psfrag{S}[c]{\small $P_4$}
\psfrag{I}[c]{\small $\bm{I_1}$}
\psfrag{J}[c]{\small $\bm{I_2}$}
    \centering
    \vspace{0pt}
    \includegraphics[width=0.9\textwidth]{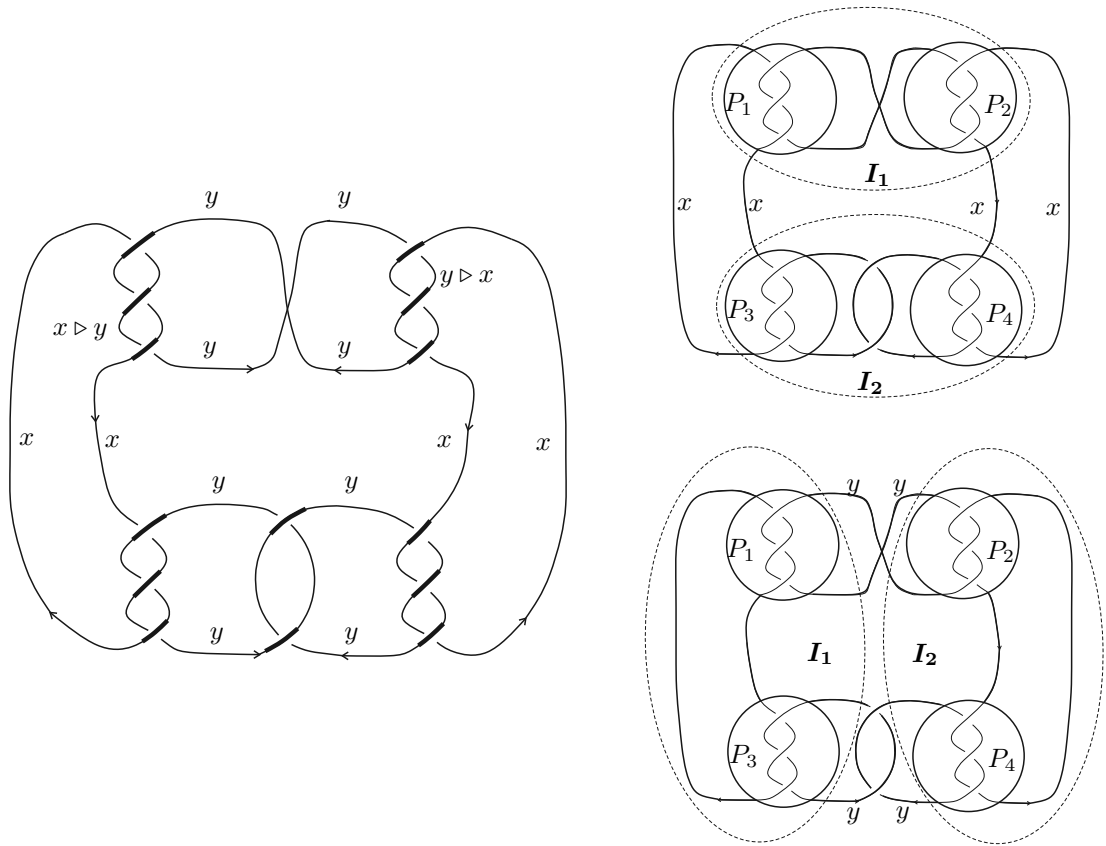}
\end{minipage}

\caption{\label{F:m3131}In this example the quandle with which the machine is coloured is commutative, \textit{i.e.} $x \trr y = y \trr x$. This machine has four non-unit irreducible factors, $P_1$ to $P_4$ (encircled on the right). It has two distinct prime factorizations: $\{P_1, P_2\}$ and $\{P_3, P_4\}$; or $\{P_1, P_3\}$ and $\{P_2, P_4\}$.}
\end{figure}%PRIME FACTORIZATION IS NOT UNIQUE

\begin{defn}[Refinement, topological equivalence]
 A \emph{refinement} $\mathcal{N}^\prime$ of a factorization $\mathcal{N}\ass N_1\begin{minipage}{8pt}\includegraphics[width=8pt]{hash.eps}\end{minipage} N_2\begin{minipage}{8pt}\includegraphics[width=8pt]{hash.eps}\end{minipage}\cdots\begin{minipage}{8pt}\includegraphics[width=8pt]{hash.eps}\end{minipage} N_k$ of machine $M$ is a factorization of $M$  obtained from $\mathcal{N}$ by factorizing one of its factors $N_i^\prime\begin{minipage}{8pt}\includegraphics[width=8pt]{hash.eps}\end{minipage} N_i^{\prime\prime} = N_i\in \mathcal{N}$ . Two factorizations which are related by a finite sequence of refinements and their inverses, via factorizations into two or more non-unit factors, are said to be \emph{topologically equivalent}.
\end{defn}

\begin{thm}[Unique prime factorization]\label{T:UniquePrimeFactorization}
Each topological equivalence class of factorizations of $M$ contains a prime factorization of $M$. Prime factorization $\mathcal{N}\ass N_1\begin{minipage}{8pt}\includegraphics[width=8pt]{hash.eps}\end{minipage} N_2\begin{minipage}{8pt}\includegraphics[width=8pt]{hash.eps}\end{minipage}\cdots \begin{minipage}{8pt}\includegraphics[width=8pt]{hash.eps}\end{minipage} N_k=M$ is unique in the following sense: If $\mathcal{N}^\prime\ass N^\prime_1\begin{minipage}{8pt}\includegraphics[width=8pt]{hash.eps}\end{minipage} N^\prime_2\begin{minipage}{8pt}\includegraphics[width=8pt]{hash.eps}\end{minipage}\cdots\begin{minipage}{8pt}\includegraphics[width=8pt]{hash.eps}\end{minipage} N^\prime_k=M$ is another prime factorization of $M$ that is topologically equivalent to $\mathcal{N}$, then there exists a permutation $\sigma$ on $k$ elements, and a set $\set{T_1,T_2,\ldots,T_k}$ of unit factors, such that $N_i= N_{\sigma(i)}^\prime\begin{minipage}{8pt}\includegraphics[width=8pt]{hash.eps}\end{minipage}\, T_i$ for all $i=1,2,\ldots,k$.
\end{thm}

Theorem~\ref{T:UniquePrimeFactorization} follows from the Diamond Lemma, whose hypotheses are satisfied by Theorem~\ref{T:Complexity} together with the following lemma.

\begin{lem}
Any two refinements $\mathcal{N}^\prime$ and $\mathcal{N}^{\prime\prime}$ of the same factorization $\mathcal{N}$ share a common refinement $\mathcal{N}^{\prime\prime\prime}$.
\end{lem}

\begin{proof}
We use the topology of sphere-and-interval tangles. Without the limitation of generality, machines are assumed to be non-split.

We first set up the necessary language.

Recall from Section~\ref{SS:MachinePrelim} that to \emph{cancel} a factor $N=(G,\bm{\phi}_H,\rho_H)$ in $M=(G,\bm{\phi},\rho)$ ($H$ denotes the domain of $\bm{\phi}_H$) is to replace $M$ by a machine $M-N\ass (G,\bm{\phi}_{G-H},\rho^H)$  where the $\rho^H$ satisfies $\rho^H(r)=\rho(r)$ for all $r\in G-H$. Topologically, we cancel a factor by replacing each of its spheres in $H$ by an interval connecting its incident segments. For concreteness, parameterizing $S^2$ as the unit sphere on the $xyz$ hyperplane in $\mathds{R}^4$, we replace $S^2$ by $(\cos(t),0,\sin(t),0)$ with $t\in[0,\pi]$, smoothing corners as required. See Figure~\ref{F:Trivialization4d}.

\begin{figure}[htb]
\includegraphics[width=4.5in]{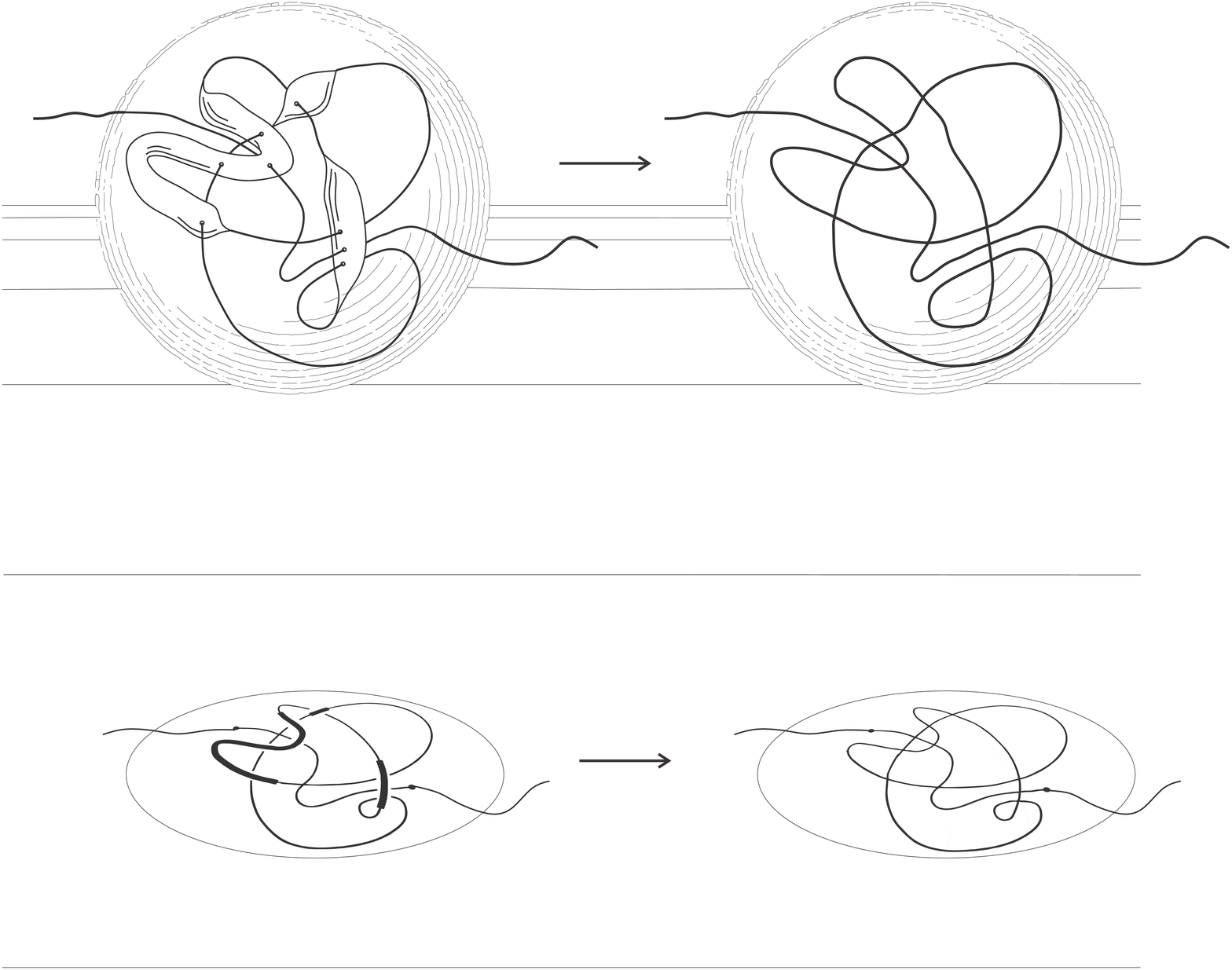}
\caption{\label{F:Trivialization4d}Cancelling a factor.}
\end{figure}

%To colour a sphere-and-interval tangle $K$, associate colours to the rackets of its universal rack, subject to the crossing condition of course.Fix a $3$--dimensional hyperplane $H$ with respect to which we take a Rosemeister diagram $D$ for $K$.

A \emph{system of decomposing spheres} for a sphere-and-interval tangle $K$ is a set of disjoint $3$--spheres $\mathcal{S}\ass\set{S_1,S_2,\ldots,S_k}$ embedded in $S^4\simeq \mathds{R}^4\cup \{\infty\}$ bounding $2k$ $4$--balls $B^\mathrm{in}_1,B^\mathrm{in}_2,\ldots,B^\mathrm{in}_k$ and $B^\mathrm{out}_1,B^\mathrm{out}_2,\ldots,B^\mathrm{out}_k$ in $S^4$, where $B^\mathrm{out}_i$ contains the point $\set{\infty}$ in its interior for $i=1,2,\ldots,k$. If $B_i$ properly contains $4$--balls $B_{\ell(1)},B_{\ell(2)},\ldots,B_{\ell(s)}$ then the \emph{domain} of $S_i$ is defined to be $B^\mathrm{in}_i$ minus the interiors of $B^\mathrm{in}_{\ell(1)},B^\mathrm{in}_{\ell(2)},\ldots,B^\mathrm{in}_{\ell(s)}$. We require that each sphere $S^i$ meets $K$ at a finite set of points (\textit{i.e.} each $S_i$ may intersect intervals of $K$, but does not intersect spheres of $K$), all of which share the same colour $x\in Q$ in a Rosemeister diagram $D_\mathcal{S}$ for $K$. We also require that $K\cap \bigcup_{i=1}^k B^{\mathrm{out}}_i$ is a unit machine, so that all of the `action' takes place inside the domains of $S_1,S_2,\ldots,S_k$. In the same vein, we also assume that $K\cap B^{\mathrm{in}}_i$ is non-unit for $i=1,2,\ldots,k$.

Lift a factorization of a machine to a system of decomposing spheres for a sphere-and-interval tangle representing it. To refine, bisect a decomposing sphere using a $3$--dimensional hyperplane $H\simeq \mathds{R}^3$, separating it into two spheres. For simplicity, we are ignoring the technical details of how to push off the resulting spheres relative to one another, smoothing corners, general position, \textit{etc.}

The factorization $\mathcal{N}$ corresponds to a set of decomposing spheres $\mathcal{S}_\mathcal{N}\ass \set{S_1,S_2,\ldots,S_{m-1}}$. If the refinements $\mathcal{N}^\prime$ and $\mathcal{N}^{\prime\prime}$ arise from bisections of distinct balls $B_i^{\mathrm{in}}$ and $B_j^{\mathrm{in}}$, we can perform both bisections simultaneously to obtain a common refinement $\mathcal{N}^{\prime\prime\prime}$ for both $\mathcal{N}^\prime$ and $\mathcal{N}^{\prime\prime}$. If both refinements are bisections of the same ball $B_{m-1}^{\mathrm{in}}$, let us take $\mathcal{S}_{\mathcal{N}^\prime}\ass\set{S_1,S_2,\ldots S_{m-2}, S^\prime_{m-1},S^\prime_{m}}$ as the system of decomposing spheres $\mathcal{N}^\prime$, and $\mathcal{S}_{\mathcal{N}^{\prime\prime}}\ass\set{S_1,S_2,\ldots S_{m-2}, S^{\prime\prime}_{m-1},S^{\prime\prime}_{m}}$ as the system of decomposing spheres $\mathcal{N}^{\prime\prime}$, where $(S_{m-1}^\prime,S_m^\prime)$ is induced by bisecting $S_{m-1}$ along a $3$--dimensional hyperplane $H^\prime$, and $(S_{m-1}^{\prime\prime},S_m^{\prime\prime})$ is induced by bisecting $S_{m-1}$ along a $3$--dimensional hyperplane $H^{\prime\prime}$.

Now that the statement of the theorem has been reformulated topologically, its proof becomes analogous to the proof of unique prime decomposition for knots (\textit{e.g} \citep{BurdeZieschang:03}). Assume general position, and cut along both $H^\prime$ and $H^{\prime\prime}$, pushing off and smoothing as required. The resulting balls are disjoint, and so there are no interactions between the factors of $M$ which they induce. Thus there exists a plane $P$ with respect to which there is a Rosemeister diagram $D_{\mathcal{N}^{\prime\prime\prime}}$ for $K$ in which the decomposing spheres appear as disjoint spheres which intersect $K$ only at segments. Cancelling the factor induced by $B_{m-1}^{\mathrm{in}}\in\mathcal{N}^{\prime}$ (the interior of $S_{m-1}^{\prime}$) does not affect the colours of the thin lines inside the projection $B_{m-1}^{\mathrm{out}}$ to $D_{\mathcal{N}^{\prime\prime\prime}}$ because that factor does not interact with these segments. Therefore, in particular, cancelling a factor of it will not effect those colours. Conversely, cancelling the factor induced by $B_{m}^{\mathrm{in}}\in\mathcal{N}^{\prime\prime}$ (the interior of $S_{m}^{\prime\prime}$) also does not affect the colours in the projection of $B_{m}^{\mathrm{out}}$ to $D_{\mathcal{N}^{\prime\prime\prime}}$, therefore in particular cancelling a factor of it will not effect those colours. Combining these two observations proves independence of the newly created factors. We are working modulo unit factors, so any unit factors that are created may be discarded. Thus, we have found the requisite common refinement.
\end{proof}

\begin{rem}
We would prefer to have an algebraic proof for Theorem~\ref{T:UniquePrimeFactorization}, but there is no such proof known even in the classical case of knots in $\mathds{R}^3$.
\end{rem}

We find prime factors in Rosemeister diagrams by trying out different systems of decomposing spheres in a Roseman diagram, and projecting those spheres down.

  What about Reidemeister diagrams? A system of decomposing spheres induces a \emph{system of cuts} for a tangle diagram of $K$, which are boundaries of discs in tangle diagrams which intersect thin strands transversely, and which may pass under over-strands. Cuts are drawn as dotted lines. A system of cuts is illustrated in Figure~\ref{F:m3131}. A system of cuts is transformed under Reidemeister moves as follows:

\begin{equation}
\begin{minipage}{4in}
\includegraphics[width=4in]{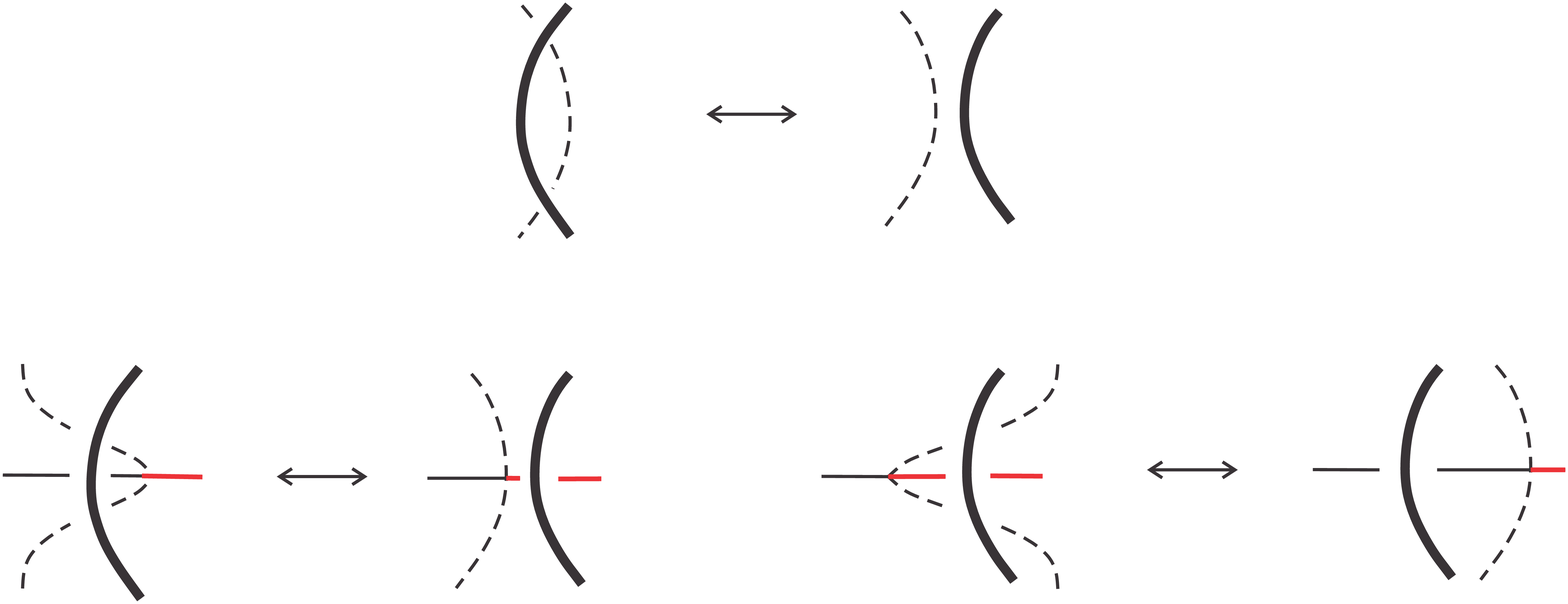}
\end{minipage}
\end{equation}

 As Figure~\ref{F:TrivBorromeanTrefoil} illustrates, it is not easy to find `good' systems of cuts for Reidemeister diagrams. The machine on the left of Figure~\ref{F:TrivBorromeanTrefoil} is an unit for any colouring, but the nontrivially coloured machines on the right are not units. They are both irreducible having a complexity of $1$. In principal, however, all factorizations do indeed arise from cut systems.

\begin{figure}
\centering
\begin{minipage}[t]{0.4\linewidth}
\centering
\vspace{0pt}
    \includegraphics[width=0.6\textwidth]{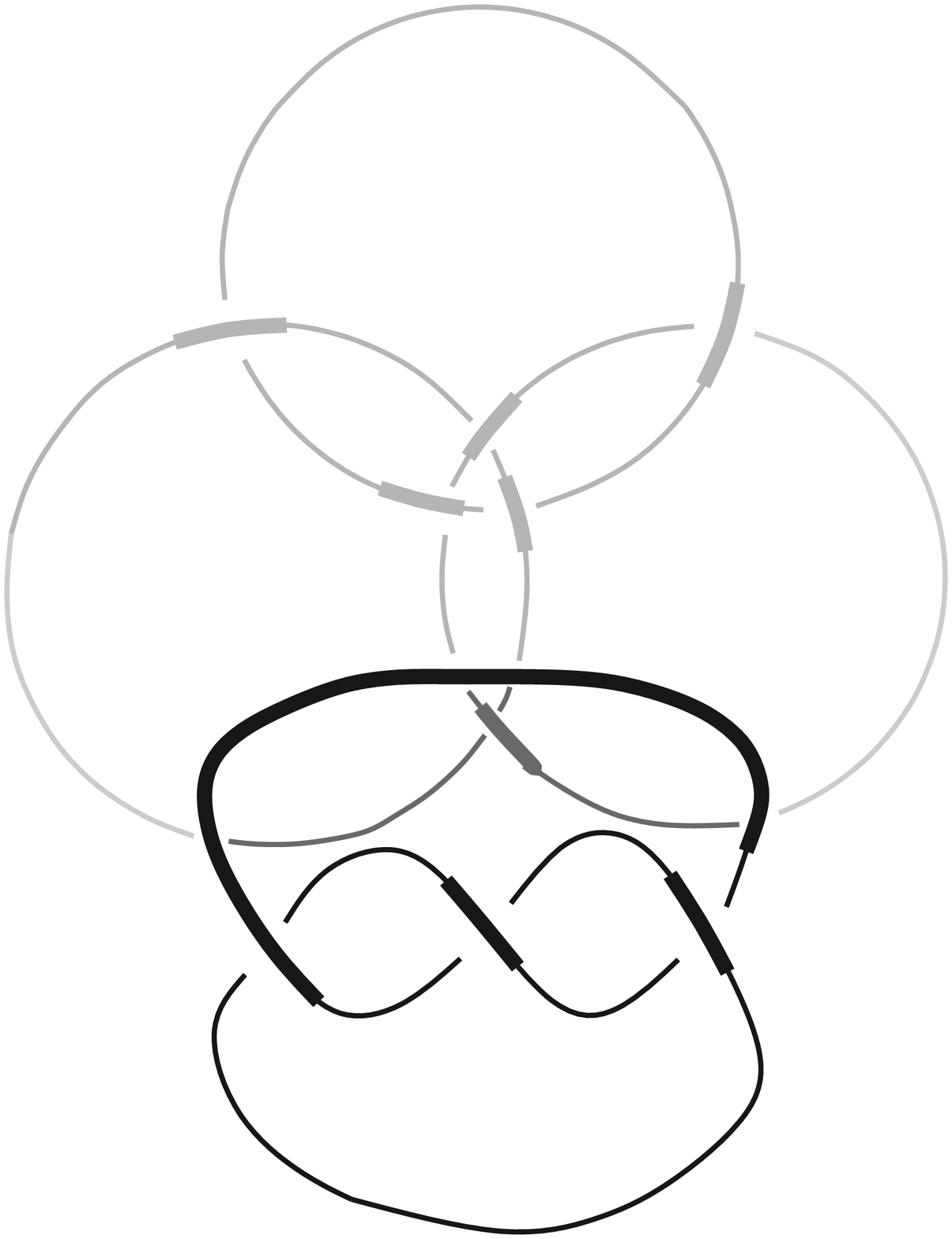}
\end{minipage}
\hspace{1mm}
\begin{minipage}[t]{0.4\linewidth}
\centering
\vspace{0pt}
    \includegraphics[width=0.6\textwidth]{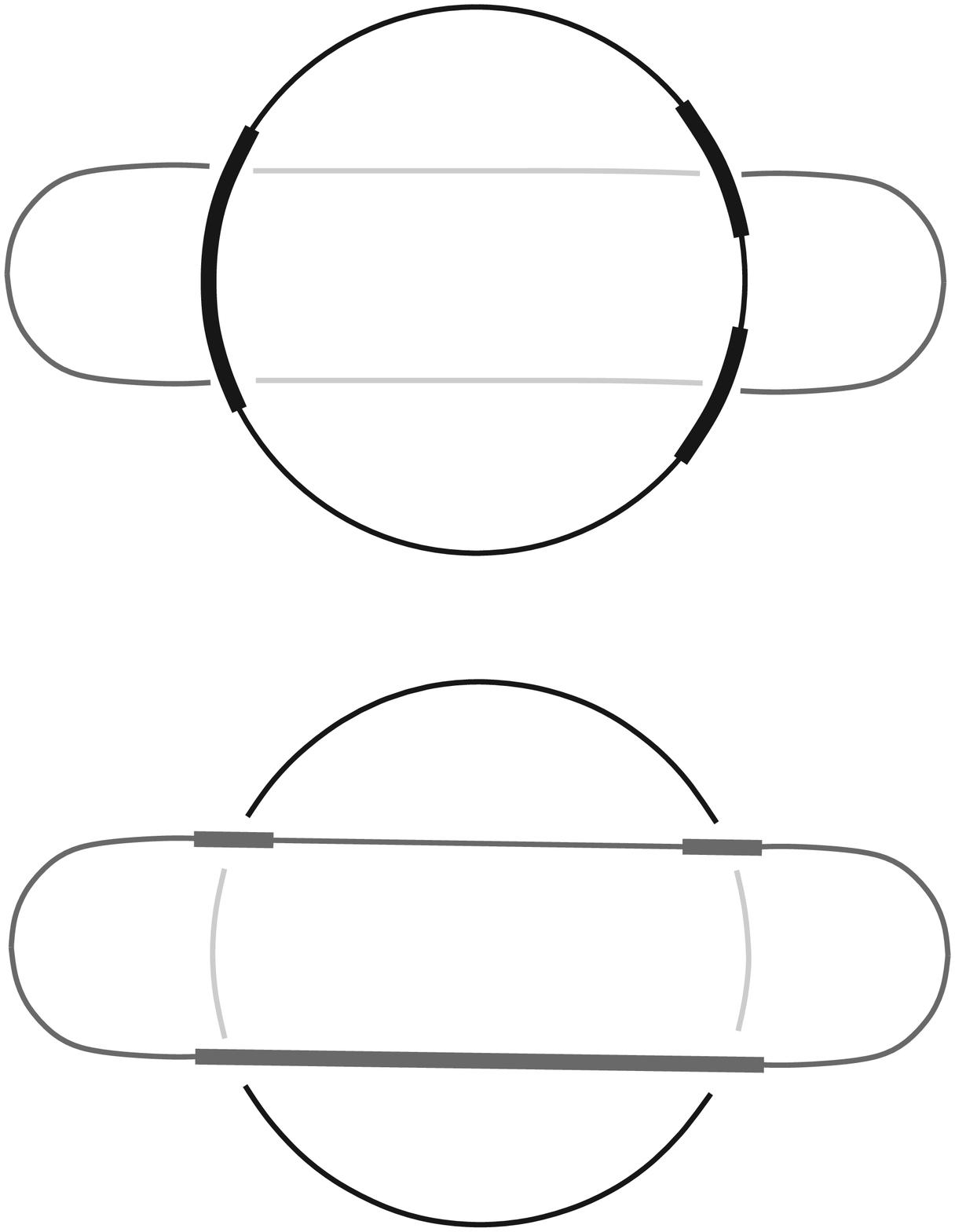}
\end{minipage}
\caption{\label{F:TrivBorromeanTrefoil}}
\end{figure}

\begin{prop}
Any factorization of a machine $M$ is induced by some system of cuts.
\end{prop}

\begin{proof}
%Recall that a \emph{system of cuts} for a tangle diagram of $K$, which are boundaries of discs in tangle diagrams which intersect thin strands of the diagram transversely, and may pass under over-strands.
By the existence of a system of decomposing spheres which by definition are disjoint, and by the Reidemeister Theorem for machines (Theorem~\ref{T:MachineReidemeisterTheorem}), $M$ is equivalent to a machine with no interactions between the factors. The proposition now follows from Theorem~\ref{T:UniquePrimeFactorization}.
 \end{proof}

\section{Conclusion}

We have exhibited colour-suppressed tangle machines as being diagrams for networks of jointly embedded spheres and intervals in standard Euclidean $\mathds{R}^4$. Our Reidemeister Theorem has demonstrated that two machines are (stably) equivalent if and only if any two sphere-and-interval tangles which they represent are (stably) equivalent.

We defined several invariants for machines:

\begin{itemize}
\item The underlying graph of a tangle machine, and its reduced version.
\item The sets of initial and terminal colours of a tangle machine.
\item The number of nontrivial interactions in a machine.
\item The fundamental rack or quandle of a machine.
\item The linking graph of a machine, and its reduced version. This contains information about relative influence of registers on processes.
\item The coloured linking graph. This contains information about the relative influence of colours of a registers on colours of processes.
\item The Shannon capacity of a machine, which measure how much information it can carry, or conversely how much information is required to encode the machine uniquely.
\item The complexity of a machine, that is its number of prime factors.
\end{itemize}

Additionally, we showed that false stabilization cannot decrease complexity, and we proved that the prime factorization of a machine is in a certain sense unique up to trivial factors. All proofs used the topological realization of a colour-suppressed machine.

\bibliographystyle{rspublicnat}
%\bibliography{tangle_machines_part1.bib}

\end{document}